\documentclass[a4paper,UKenglish,cleveref, autoref, thm-restate]{lipics-v2021}
\usepackage{graphicx} 
\usepackage{quiver}
\usepackage{amsfonts}
\usepackage{amsmath}
\usepackage{amsthm}
\usepackage[disable]{todonotes}
\usepackage{amssymb}
\usepackage{mathtools}
\usepackage[capitalise]{cleveref}
\usepackage{float}
\usepackage{graphicx}

\usepackage{xcolor}
\usepackage{listings}

\definecolor{dkviolet}{RGB}{102, 0, 153}
\definecolor{ltblue}{RGB}{0, 0, 0}
\definecolor{dkblue}{RGB}{0, 0, 139}
\definecolor{dkgreen}{RGB}{0, 100, 0}
\lstdefinelanguage{Coq}{ 
    mathescape=true,
    texcl=false, 
    morekeywords=[1]{Section, Module, End, Require, Import, Export,
        Variable, Variables, Parameter, Parameters, Axiom, Hypothesis,
        Hypotheses, Notation, Local, Tactic, Reserved, Scope, Open, Close,
        Bind, Delimit, Definition, Let, Ltac, Fixpoint, CoFixpoint, Add,
        Morphism, Relation, Implicit, Arguments, Unset, Contextual,
        Strict, Prenex, Implicits, Inductive, CoInductive, Record,
        Structure, Canonical, Coercion, Context, Class, Global, Instance,
        Program, Infix, Theorem, Lemma, Corollary, Proposition, Fact,
        Remark, Example, Proof, Goal, Save, Qed, Defined, Hint, Resolve,
        Rewrite, View, Search, Show, Print, Printing, All, Eval, Check,
        Projections, inside, outside, Def},
    morekeywords=[2]{forall, exists, exists2, fun, fix, cofix, struct,
        match, with, end, as, in, return, let, if, is, then, else, for, of,
        nosimpl, when},
    morekeywords=[3]{Type, Prop, Set, true, false, option},
    morekeywords=[4]{pose, set, move, case, elim, apply, clear, hnf,
        intro, intros, generalize, rename, pattern, after, destruct,
        induction, using, refine, inversion, injection, rewrite, congr,
        unlock, compute, ring, field, fourier, replace, fold, unfold,
        change, cutrewrite, simpl, have, suff, wlog, suffices, without,
        loss, nat_norm, assert, cut, trivial, revert, bool_congr, nat_congr,
        symmetry, transitivity, auto, split, left, right, autorewrite},
    morekeywords=[5]{by, done, exact, reflexivity, tauto, romega, omega,
        assumption, solve, contradiction, discriminate},
    morekeywords=[6]{do, last, first, try, idtac, repeat},
    morecomment=[s]{(*}{*)},
    showstringspaces=false,
    morestring=[b]",
    morestring=[d]’,
    tabsize=3,
    extendedchars=false,
    sensitive=true,
    breaklines=false,
    basicstyle=\small,
    captionpos=b,
    columns=[l]flexible,
    identifierstyle={\ttfamily\color{black}},
    keywordstyle=[1]{\ttfamily\color{dkviolet}},
    keywordstyle=[2]{\ttfamily\color{dkgreen}},
    keywordstyle=[3]{\ttfamily\color{ltblue}},
    keywordstyle=[4]{\ttfamily\color{dkblue}},
    keywordstyle=[5]{\ttfamily\color{dkred}},
    stringstyle=\ttfamily,
    commentstyle={\ttfamily\color{dkgreen}},
    literate=
    {\\forall}{{\color{dkgreen}{$\forall\;$}}}1
    {\\exists}{{$\exists\;$}}1
    {<-}{{$\leftarrow\;$}}1
    {=>}{{$\Rightarrow\;$}}1
    {==}{{\code{==}\;}}1
    {==>}{{\code{==>}\;}}1
    {->}{{$\rightarrow\;$}}1
    {<->}{{$\leftrightarrow\;$}}1
    {<==}{{$\leq\;$}}1
    {\\o}{{$\circ\;$}}1 
    {\@}{{$\cdot$}}1 
    {\/\\}{{$\wedge\;$}}1
    {\\\/}{{$\vee\;$}}1
    {++}{{\code{++}}}1
    {~}{{$\sim$}}1
    {\@\@}{{$@$}}1
    {\\mapsto}{{$\mapsto\;$}}1
    {\\hline}{{\rule{\linewidth}{0.5pt}}}1
}[keywords,comments,strings]

\newcommand{\id}{\textrm{id}}

\newcommand{\transform}[1]{\stackrel{#1}{\longrightarrow}}
\newcommand{\transformv}[1]{\stackrel{#1}{\longmapsto}}
\newcommand{\Yes}{\textrm{Yes}}
\newcommand{\No}{\textrm{No}}
\newcommand{\counit}{\varepsilon}
\newcommand{\comma}{\textrm{,}}
\newcommand{\putf}{\mathtt{put}}
\newcommand{\strength}{\mathtt{strength}}
\newcommand{\ctx}{\mathtt{ctx}}
\newcommand{\concat}{\mathtt{concat}}
\renewcommand{\prod}{\alpha}
\newcommand{\catname}[1]{{\normalfont\textbf{#1}}}
\newcommand{\Set}{\catname{Set}}

\newcommand{\smallpicc}[1]{\begin{figure}[H]\centering\includegraphics[width=0.3\textwidth]{#1.pdf} \end{figure}}
\newcommand{\custompicc}[2]{\begin{figure}[H]\centering\includegraphics[width=#2\textwidth]{#1.pdf} \end{figure}}
\newcommand{\PL}{\overrightarrow{L}}
\newcommand{\downmapsto}{\rotatebox[origin=c]{-90}{$\scriptstyle\mapsto$}\mkern2mu}

\theoremstyle{plain}
\newtheorem{innerreplemma}{Lemma}
\newenvironment{replemma}[1]
{\renewcommand{\theinnerreplemma}{#1}\innerreplemma}
{\endinnerreplemma}

\title{Monads, Comonads, and Transducers} 

\author{Rafa\l \ Stefa\'nski}{University of Warsaw, Poland}{rafal.stefanski@mimuw.edu.pl}{https://orcid.org/0000-0002-8439-4056}{}
\authorrunning{R. Stefanski} 

\Copyright{Rafa\l \ Stefa\'nski} 



\category{} 

\relatedversion{} 

\nolinenumbers 
\hideLIPIcs

\EventEditors{John Q. Open and Joan R. Access}
\EventNoEds{2}
\EventLongTitle{42nd Conference on Very Important Topics (CVIT 2016)}
\EventShortTitle{CVIT 2016}
\EventAcronym{CVIT}
\EventYear{2016}
\EventDate{December 24--27, 2016}
\EventLocation{Little Whinging, United Kingdom}
\EventLogo{}
\SeriesVolume{42}
\ArticleNo{23}

\begin{document}

\maketitle

\begin{abstract}
This paper proposes a definition of \emph{recognizable transducers over monads and comonads},
which bridges two important ongoing efforts in the current research on regularity. 
The first effort is the study of regular transductions, which extends the notion
of regularity from languages into word-to-word functions. The other important effort 
is generalizing the notion of regular languages from words to arbitrary monads, 
introduced in \cite{bojanczyk2015recognisable} and further developed
in \cite{bojanczyk2020languages, bojanczyk2023monadic}. In the paper,
we present a number of examples of transducer classes that fit the proposed framework.
In particular we show that our class generalizes the classes of Mealy machines and rational transductions. 
We also present examples of recognizable transducers for infinite words and a specific type of trees called terms. 
The main result of this paper is a theorem, which states the class of
recognizable transductions is closed under composition, subject to some \emph{coherence} axioms 
between the structure of a monad and the structure of a comonad.
Due to its complexity, we formalize the proof of the theorem in \emph{Coq Proof Assistant} \cite{Coq-refman}.
In the proof, we introduce the concepts of a \emph{context} and a
\emph{generalized wreath product} for Eilenberg-Moore algebras,
which could be valuable tools for studying these algebras.
\end{abstract}

\section{Introduction}
The study of transductions plays an important role in studying and understanding the theory of regularity.
Although this idea is not new, and its importance has been known for decades
(see the first paragraph of \cite[Section~V]{scott1967some}),
it seems to have been gaining momentum in the recent years (e.g. see 
\cite{colcombet2023integer, dolatian2019learning, filiot2022rational, dartois2017reversible}). \todo{citations}
This paper aims to extend the concept of \emph{languages recognized over monads},
introduced by \cite{bojanczyk2015recognisable}, to transductions. 
Interestingly, this requires studying functors that are simultaneously equipped 
with the structures of both a monad and a comonad.

Although, this work is clearly inspired by category theory, our primary focus 
lies within the domain of formal languages and transducers. 
For this reason, we do not assume any prior knowledge of category theory 
on the part of the reader. We will provide all necessary definitions 
and limit our discussion to the basic category \Set, 
which consists of sets and functions between them. For a discussion 
about extending this work to other categories, see \cref{subsec:ccc}.  

The paper is structured as follows: In \cref{sec:monads-and-languages}, we summarize
the results on \emph{languages recognizable over monads} (based on \cite{bojanczyk2015recognisable,bojanczyk2020languages,bojanczyk2023monadic}), 
which serves as a context for this paper.
In \cref{sec:comonads-and-transdcuers}, we introduce \emph{transductions recognizable over monads and comonads}, which is the 
main contribution of this paper. In \cref{sec:compositions-of-recognisable-transductions}, we show that 
our proposed classes of transducers are closed under compositions, 
subject to certain \emph{coherence axioms}.  This serves two purposes:
First, it serves as a validation for our class of transductions. Second, 
it facilitates a deeper understanding of the structure of the monad/comonad functors and 
their Eilenberg-Moore algebras. 
Finally, in \cref{sec:futher-work}, we outline the potential directions for further work.

Let us also  mention that some of the proofs (including the proof of the compositions theorem)
are verified in the \emph{Coq Proof Assistant} \cite{Coq-refman} (see \cref{subsec:coq}).

\paragraph*{Related work}
A recently published work \cite{bojanczyk2023algebraic} also presents a
categorical framework for defining transductions over monads. 
Our paper differs from \cite{bojanczyk2023algebraic} in two important ways:
First, \cite{bojanczyk2023algebraic} focuses on generalizing the class of regular functions (i.e. those recognized by two-way transducers),
whereas this paper concentrates on generalizations of Mealy machines and rational functions (see also Item~1 in Section~\ref{sec:futher-work}).
Second, the two papers present different approaches to the problem.
In particular, our paper develops a comonadic framework that 
does not appear in \cite{bojanczyk2023algebraic}.
We believe that both the approaches are valuable and warrant further investigation.
Future research could explore potential connections and synergies between the two methodologies.

\section{Monads and recognizable languages}\label{sec:monads-and-languages}
In this section we present a brief summary of the existing research on \emph{recognizable languages over monads}, 
which is the starting point for this paper. This line of research was initiated in \cite{bojanczyk2015recognisable},
and then continued in \cite{bojanczyk2020languages} and in \cite{bojanczyk2023monadic}. The main idea is to 
approach regular languages from the algebraic perspective, and then use Eilenberg-Moore algebras to generalize their definition from languages over words to languages 
over arbitrary monads (such as infinite words, trees, or even graphs).
We start the summary with the following, well-established definition of monoid recognizability:
\begin{definition}
\label{def:monoid-recognisable}
    A monoid is a set equipped with a transitive binary operation and an identity element.
    We say that a language $L \subseteq \Sigma^*$ is \emph{recognizable by a monoid} if there exist: 
    \[
        \begin{tabular}{ccc}
            $\underbrace{M}_{\textrm{Monoid}}$ & 
            $\underbrace{h : \Sigma \to M}_{\textrm{Input function}}$ & 
            $\underbrace{F \subseteq M}_{\textrm{Accepting set}}$,
        \end{tabular}
    \]
    such that a word $w_1 \ldots w_n \in \Sigma^*$ belongs to $L$, if and only if:
    \[h(w_1) \cdot h(w_2) \cdot \ldots \cdot h(w_n) \in F \]
\end{definition}
It is a well-know fact, that the class of languages that can be recognized by \emph{finite} 
monoids is equal to regular languages (see \cite[Theorem 3.21]{pin2010mathematical} for details). 

As mentioned before, the key idea presented by \cite{bojanczyk2015recognisable} is extending 
\cref{def:monoid-recognisable} from languages of words, to languages of arbitrary monads. 
Before we show how to do this, we need to define \emph{monads}. 
However, since monads are a special kind of functors, we need to start by 
giving a definition of a \emph{functor}\footnote{We only define a special case endofunctors in $\Set$, as those are the
only type of functors that we are going to use. See \cite[Definition~41]{abramsky2011introduction} for a general definition.}:
\begin{definition}\label{def:functor}
    A \emph{functor} $M$ consists two parts. The first part is a mapping from sets to sets, 
    i.e. for every set $X$, the functor $M$ assigns another set denoted as $M X$. 
    The other part is a mapping from functions to functions, i.e. for every function $f : X \to Y$, 
    the functor $M$ assigns a function $M f : M X \to M Y$. Moreover, the function mapping needs to 
    satisfy the following axioms (where $\id_X$ is the identity function on $X$, and $\circ$ is function composition):
    \[
        \begin{tabular}{cc}
            $M \id_X = \id_{M X}$ & $M (f \circ g) = (M f) \circ (M g)$,
        \end{tabular} 
    \]
\end{definition}
For example,  let us show how to apply this definition to \emph{finite lists}:
\begin{example}\label{ex:finite-lists}
    The \emph{finite lists} functor maps every set $X$ into $X^*$ (i.e. the set of finite lists over $X$), 
    and it maps every function $f : X \to Y$ into a function $f^* : X^* \to Y^*$ that applies $f$ element
    by element. It is not hard to see that this satisfies the axioms from \cref{def:functor}. 
\end{example}
We are now ready to give the definition of a \emph{monad}:
\begin{definition}\label{def:monad}
    A \emph{monad} is a functor $M$ equipped with two operations;
    \[
        \begin{tabular}{ccc}
                $\eta_X : X \to M X$ & and & $\mu_X :  M M X \to X$  
        \end{tabular}
    \]
    We are going to refer to $\eta$ as the \emph{singleton} operation,
    and to $\mu$ as the \emph{flatten} operation. The two operations need
    to satisfy the axioms of a monad, which can 
    be found in Section~\ref{subsec:monad-axioms} of the appendix. 
\end{definition}
Let us now show how to apply this definition to \emph{finite lists}:
\begin{example}\label{ex:monad-list}
    The functor of \emph{finite lists} can be equipped with the following monad structure. 
    The singleton operation $\eta : X \to X^*$ returns a singleton list with the argument,
    and the flatten operation $\mu : M M X \to M X$ flattens the list of lists into a single list.
    For example:
    \[ 
        \begin{tabular}{ccc}
            $\eta(3) = [3]$ & and & $\mu([[1, 2, 3], [4, 5, 6], [], [7,8], [9]]) = [1, 2, 3, 4, 5, 6, 7, 8, 9]$
        \end{tabular}
    \]
\end{example} 
From the perspective of \cite{bojanczyk2015recognisable}, the most important feature 
of monads is that they can be used to define \emph{Eilenberg-Moore algebras}, 
which can be seen as a generalization of monoids:
\begin{definition}
\label{def:em-algabra}
    An \emph{Eilenberg-Moore algebra} for a monad $M$ is a set $S$ together 
    with a multiplication function $\prod : M S \to S$, that makes the
    following diagrams commute\footnote{
        We hope that the notation of commutative diagrams is intuitively clear.
        See \cite[Section~1.1.3]{abramsky2011introduction} for more explanation.
     }:
\[\begin{tikzcd}
	{M \, M \, S} &&& {M \, S} && S && {M\, S} \\
	\\
	{M \, S} &&& S &&&& S
	\arrow["{\mu_S}", from=1-1, to=3-1]
	\arrow["{M \, \prod}"', from=1-1, to=1-4]
	\arrow["\prod", from=3-1, to=3-4]
	\arrow["\prod"', from=1-4, to=3-4]
	\arrow["{\textrm{id}_S}"{description}, from=1-6, to=3-8]
	\arrow["{\eta_S}", from=1-6, to=1-8]
	\arrow["\prod", from=1-8, to=3-8]
\end{tikzcd}\]
\end{definition}

To understand the intuition behind this definition, let us show that there is a 
bijective correspondence between Eilenberg-Moore algebras for finite lists and monoids:
\begin{example}
\label{ex:monoids-as-em-algebras}
   Let $(X, \prod)$ be an Eilenberg-Moore algebra for the finite list monad. We can use $\prod$ to 
   define a binary operation and an identity element on $X$, obtaining a monoid structure
   on $X$:
   \[ \begin{tabular}{ccc}
    $ x \cdot y = \prod([x, y])$ & and & $1_X = \prod([\,])$,
   \end{tabular} \]
   To see that this is a valid monoid observe that:
   \[ x \cdot (y \cdot z) \stackrel{\textrm{def}}{=} \prod([x, \prod([y, z])]) \stackrel{\textrm{Ax.} \eta}{=} \prod([\prod([x]), \prod([y, z])])
   \stackrel{\textrm{Ax.} \mu}{=} \prod(\mu [[x], [y, z]]) = \prod([x, y, z]) \textrm{,} \]
   where Ax.$\mu$ and Ax.$\eta$ denote the Eilenberg-Moore axioms from \cref{def:em-algabra}.
   Using a similar reasoning one can show that $(x \cdot y) \cdot z = \prod([x, y, z])$, which 
   means that the new binary operation is associative. Similarly, one can show that $1_X$
   is indeed an identity element.

   To see that this defines a bijection between monoids and Eilenberg-Moore algebras for finite lists,
   we define an inverse mapping that defines $\alpha$ in terms of the monoid structure:
   \[ \begin{tabular}{ccc}
    $\prod([x_1, x_2, \ldots, x_n]) = x_1 \cdot x_2 \cdot \ldots \cdot x_n$ & and & $\prod([\, ]) = 1_X$ 
   \end{tabular}
   \]
\end{example}
We are now ready to present the key definition from \cite{bojanczyk2015recognisable},
i.e. a \emph{recognizable language for a monad $M$}:
\begin{definition}\label{def:recognisable-language}
Let $M$ be a monad. We define an $M$-language over an alphabet $\Sigma$ 
to be a subset of $M \Sigma$.  We say that a language $L
\subseteq M \Sigma$ is $M$-definable if there exist: 
\[
    \begin{tabular}{ccc}
        $\underbrace{(A, \prod)}_{\textrm{Finite } M\textrm{-algebra}}$ & 
        $\underbrace{h : \Sigma \to A}_{\textrm{Input function}}$ &
        $\underbrace{\lambda : \Sigma \to \{ \Yes, \No \}}_{\textrm{Acceptance function}}$,
    \end{tabular}
\]
such that the characteristic function of $L$ is equal to the following composition:
\[ M \Sigma \transform{M \, h} M A \transform{\prod} A \transform{\lambda} \{ \Yes, \No \} \]
\end{definition}
Thanks to \cref{ex:monoids-as-em-algebras}, it is not hard to see that for $M$ equal to
the monad of finite lists, the definition of $M$-recognizability is equivalent to \cref{def:monoid-recognisable}, 
which further means that finite-list recognizability is equivalent to regularity. 
There are many other examples of monads, which define important classes of $M$-definable languages 
(e.g. see \cite[Section 4.3]{bojanczyk2020languages}). In this paper, let us define two more,
the \emph{countable order monad} and the \emph{terms monad}:

\begin{example}[{\cite[Example~15]{bojanczyk2015recognisable}}]
\label{ex:countable-orders-monads}
Define a \emph{countable chain} over a set $X$, to be a countable linear order 
where every position is labelled with an element of $X$.
We say that two chains are equal if there is an isomorphism between them that 
preserves both the order and the labels\footnote{
    It is worth noting that the set of all countable sets is not, strictly speaking, a set. 
    However, because we equate all chains modulo isomorphism, the set of all 
    chains over any given $X$ forms a set. 
    This is because, we can select an arbitrary infinite countable set
    and assume that all positions in every chain are elements of that set.
}.
Let us denote the set of all countable chains over $X$ as $C X$, and 
let us show that $C$ is a monad.
First, let us define the functor structure on $C$. 
If $f$ is a function of type $X \to Y$, 
then $C f : C X \to C Y$ is a function that applies $f$ to every label of a chain
(and does not modify the linear order). This leaves with
defining the monad structure: The \emph{singleton} operation $\eta : X \to C X$ 
returns a one element chain, whose only position is labelled by the input letter.
The \emph{flatten} operation $\mu_X : C C X \to C X$ is defined in terms of the lexicographic product:
Let $w : C (C X)$ be a chain of chains, then $\mu w$ is defined as follows: 
\begin{enumerate}
    \item Its positions are pairs $(x, y)$ where $x$ is a position in $w$ and 
          $y$ is a position in $w_x$ (i.e. the label of $x$ in $w$);
    \item The label of $(x, y)$ is equal to the label of $y$;
    \item The pairs are ordered lexicographically, i.e. $(x_1, y_1) \leq (x_2, y_2)$ if 
            $x_1 < x_2$ or if $x_1 = x_2$ and $y_1 \leq y_2$.
\end{enumerate}
It is not hard to verify that this construction satisfies the monad axioms
(see \cite[Example~15]{bojanczyk2015recognisable} for details).
It follows that we can apply \cref{def:recognisable-language} and 
to define the class of $C$-recognizable languages. 
Interestingly, it turns out that this class is equal to 
the class of languages that can be recognized by the \textsc{mso}-logic
(see \cite[Theorem~5.1]{carton2018algebraic} or \cite[Theorems~3.12~and~3.21]{bojanczyk2020languages} for details), 
which is a strong argument for the intuition that $C$-recognizability corresponds to the intuitive notion of regularity. 
In particular, if we only consider those $C$-recognizable languages that only contain $\omega$-words,
we obtain the class of $\omega$-regular languages. 
\end{example}

\begin{example}[{\cite[Example~21]{bojanczyk2020languages}}]
\label{ex:term-monad}
    Fix a ranked set $\mathcal{S}$, i.e. a set where every element has an associated
    arity from the set $\{0, 1, 2, \ldots \}$. 
    For every set $X$, we define $T_{\mathcal{S}} X$ to be the set of all finite trees 
    where all leaves are labelled with elements from $X$, 
    and all inner nodes are labelled with elements from $\mathcal{S}$, 
    in such a way that the number of children of each inner node is equal to 
    the arity of its label. Here is an example for $ X = \{x, y, z\}$ and
    $\mathcal{S} = \{ \underbrace{a}_{\textrm{arity 2}}, \underbrace{b}_{\textrm{arity 1}}, \underbrace{c}_{\textrm{arity 0}}  \}$:
    \custompicc{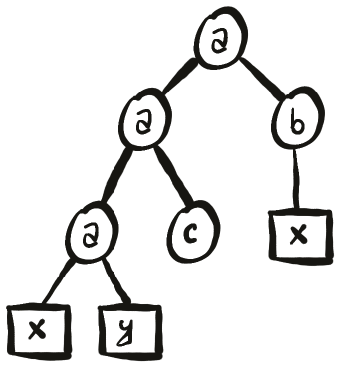}{0.2}
    \noindent (Observe that the leaves labelled with elements of arity $0$ from $\mathcal{S}$ are treated as inner nodes.)
    \noindent The set $T_\mathcal{S} X$ can also be seen as the set of terms over the signature $\mathcal{S}$, 
    with variables from~$X$. For example, for the following $\mathcal{S}$, 
    elements of $T_\mathcal{S}X$ are terms of propositional logic:
    \[ \{ \underbrace{\vee}_{\textrm{arity 2}}, \underbrace{\wedge}_{\textrm{arity 2}},
    \underbrace{\neg}_{\textrm{arity 1}}, \underbrace{\mathtt{true}}_{\textrm{arity 0}}, \underbrace{\mathtt{false}}_\textrm{arity 0} \} \]
    For every fixed $\mathcal{S}$, we define a monad structure for $T_\mathcal{S}$ called \emph{the term monad}.
    The function mapping $T_\mathcal{S} f$ applies $f$ to every leaf, and does not modify the inner nodes.
    The \emph{singleton} operation returns a tree that consists of a single leaf, labelled by the input argument. 
    Finally, the \emph{flatten} operation simply unpacks the trees from the leaves, as presented in the following figure:
    \custompicc{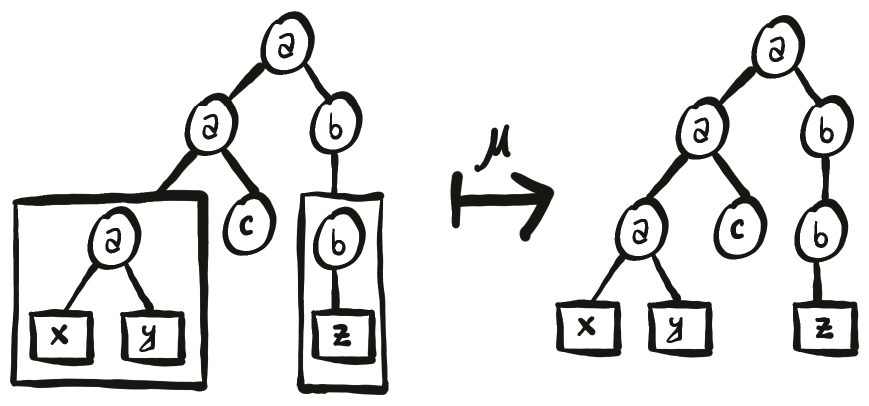}{0.5}
    \noindent
    Since $T_\mathcal{S}$ is a monad, we can use \cref{def:recognisable-language} and define the class of $T_\mathcal{S}$-recognizable languages. 
    It turns out that this class coincides with the usual notion of regularity for finite-tree languages -- 
    this is because $T_\mathcal{S}$-algebras turn out to be practically the same as deterministic bottom-up tree automata 
    (but without distinguished initial and accepting states -- which are replaced by the input and output functions from 
    \cref{def:recognisable-language}). See \cite[Example~21]{bojanczyk2020languages} for details.
\end{example}

Examples \ref{ex:countable-orders-monads} and \ref{ex:term-monad} (together with all the examples from \cite[Sections~4.2,~5.1, and~6.2]{bojanczyk2015recognisable}) 
show that \cref{def:recognisable-language} is abstract enough to capture the notion of regularity for many different types of objects. 
On the other hand, it is also concrete enough to allow for an interesting general theory of $M$-recognizability.
Examples of theorems include existence of syntactic algebras (\cite[Theorem~4.19]{bojanczyk2020languages}) and
equivalence of algebra and language varieties (\cite[Theorem~4.26]{bojanczyk2020languages}). 
There is also an ongoing quest for relating \textsc{mso}-definability and $M$-recognizability 
(see \cite{bojanczyk2023monadic}). 

\section{Comonads and transducers}\label{sec:comonads-and-transdcuers}
In this section, we extend the theory of $M$-recognizability to transducers. We start 
by presenting the definition of a \emph{comonad}, which is the dual notion to a monad.
Then, we show that for every functor $M$, that is both a monad and a comonad, 
we can define the notion of $M$-recognizable
transductions. Finally, we provide some examples of such functors $M$, and discuss their $M$-recognizable transduction classes.
Let us start with the definition:
\begin{definition}
\label{def:comonad}
        A \emph{comonad} is a functor $M$ equipped with two operations:
        \[
            \begin{tabular}{ccc}
                    $\counit_X : M X \to X$ & and & $\delta_X :  M X \to  M M X$  
            \end{tabular}
        \]
        We are going to refer to $\counit$ as the \emph{extract} operation,
        and to $\delta$ as the \emph{expand} operation. 
        The axioms of a  comonad are dual to the axioms of a monad, and can be found in Section~\ref{subsec:comonad-axioms} of the appendix. 
\end{definition} 

\begin{example}\label{ex:non-empty-lists-comonad}
For example, let us consider the functor $X^+$ of non-empty lists (where the lifting 
operation $f^+$ is defined as applying $f$ to every element of the list).
For this functor, we can define the comonad structure in the following way.
The \emph{extract} operation returns the last element of a list, and the \emph{extend} operation transforms a list
into a list of all its prefixes:
\[
    \begin{tabular}{cc}
        $\counit([x_1, \ldots, x_n]) = x_n$ &  $\delta([x_1, \ldots, x_n]) = [[x_1], [x_1, x_2],  \ldots, [x_1, \ldots x_n]]$
    \end{tabular} 
\]
(Observe that the definition of $\varepsilon$ crucially depends on the input being a non-empty list.)
It is not hard to verify that such $\delta$ and $\varepsilon$ satisfy the axioms of a comonad.

This is not the only comonad structure one can define for $X^+$. Symmetrically,
\emph{extract} could return the first element, and \emph{extend} could compute the list of all suffixes:
\[
    \begin{tabular}{cc}
        $\counit([x_1, \ldots, x_n]) = x_1$ &  $\delta([x_1, \ldots, x_n]) = [[x_1, \ldots, x_n], [x_2, \ldots, x_{n}], \ldots, [x_n]]$
    \end{tabular} 
\]
To tell those two comonads apart, we denote the \emph{prefix comonad} as $\overrightarrow{L} X$, 
and the \emph{suffix comonad} as $\overleftarrow{L} X$.
\end{example}

Next, let us observe that $X^+$ also exhibits the structure of a monad. The structure is the same as the one for $X^*$ --
the \emph{flatten} operation flattens a list into a list of lists (note that this preserves nonemptiness), 
and the \emph{unit} operation returns a singleton list.
This means that functors $\overrightarrow{L} X$  and $\overleftarrow{L} X$ are at the same time both monads and comonads.
The key observation made in this paper is that for such functors one can define a natural class of transductions:
\begin{definition}
    \label{def:recognisable-transducers}
    Let $M$ be a functor that is both a monad and a comonad. We say that a function $f : M \Sigma \to M \Gamma$ 
    is an $M$-recognizable transduction if there exist: 
    \[
        \begin{tabular}{ccc}
            $\underbrace{(A, \prod)}_{\substack{
                \textrm{a finite } M\textrm{-algebra}\\
                (\textrm{with respect the monad structure of } M)
            }}$  &
            $\underbrace{h : \Sigma \to A}_{\textrm{an input function}}$ &
            $\underbrace{\lambda : A \to \Gamma}_{\textrm{an output function}}$,
        \end{tabular}
    \]
    such that $f$ is equal to the following composition: 
    \[ M \,\Sigma \transform{M\, h} M \, A \transform{\delta} M\, M\, A \transform{M\, \prod} M\, A \transform{M\, \lambda} M\, \Gamma\]
\end{definition}

As it turns out, many interesting classes of transductions can be defined as $M$-recognizable transducers 
for a suitable $M$. In Subsections~\ref{subsec:word-to-word-transducer-functors}~and~\ref{subsec:other-examples}, we present some examples.

\subsection{Examples of word-to-word transductions}
\label{subsec:word-to-word-transducer-functors}
We start by studying $\overrightarrow{L}$-recognizable transductions.
After unfolding the definition, we obtain that each such transduction is defined by
a finite semigroup\footnote{A semigroup is a monoid that might not have the identity element. A reasoning 
similar to \cref{ex:monoids-as-em-algebras} demonstrates that semigroups
are the Eilenberg-Moore algebras for the monad $X^+$.} $S$, an input function $h : \Sigma \to S$,
and an output function $\lambda : S \to \Gamma$. The transduction is then defined in the following way:
\[ \begin{tabular}{ccc}
    $\underbrace{a_1\ a_2\ \ldots\ a_n}_{\Sigma^+} $ & $\mapsto$ & $\underbrace{\lambda\left( h(a_1)\right),\  \lambda\left(h(a_1) \cdot h(a_2)\right),\ \ldots,\  \lambda\left(h(a_1) \cdot \ldots \cdot h(a_n)\right)}_{\Gamma^+}$
\end{tabular}
\]  
In other words, we obtain the $i$-th letter of the output, by taking the 
$i$-th prefix of the input, computing the $S$-product of its $h$-values, 
and applying the output function $\lambda$. By comparing this with
\cref{def:monoid-recognisable}, we see that this means that the 
$i$-th letter of the output is computed based on regular properties of the $i$-th prefix.
This means that $\overrightarrow{L}$-recognizable transductions are equivalent 
to a well-studied transducer model called \emph{Mealy machines}. (See \cref{subsec:mealy} for 
the exact definition of Mealy machines and the proof of equivalence.)
Similarly, one can show that the class of $\overleftarrow{L}$-recognizable transductions
is equivalent to the right-to-left variant of Mealy machines.

Next, we consider \emph{length-preserving rational functions}, 
which can be defined as the class of transductions recognized by unambiguous Mealy machines
(see \cref{subsec:rational-lenght-preserving} for the definition), or more abstractly
as the class of transductions where the $i$-th letter of the output depends
on the $i$-th letter of the input and on regular properties of the $(i-1)$-st prefix and $(i+1)$-st suffix. 
Below we define a monad/comonad functor of \emph{pointed list}\footnote{
    This is a well-known functor
        and both its monad \cite[Section~8]{bojanczyk2015recognisableFull} and its comonad \cite[Section 3.2.3]{orchard2014programming} structures have been studied in the past.
        (Although, rarely together.)
} that recognizes this class:
\begin{definition}
    We define $\overline{L} X$ to be the set of all non-empty lists over $X$ with exactly one underlined element.
    For example, $[a, \underline{b}, b, c]$ is an element of  $\overline{L} \{a, b, c\}$. 
    This is clearly a functor, with $\overline{L} f$ defined as simply applying $f$ to every element of the list (and 
    keeping the underline where it was). The monadic and comonadic operations work as follows: 
    \begin{enumerate}
        \item The \emph{singleton} operation returns a list with one underlined element, e.g.: $\eta(a) = [\underline{a}]$. 
        \item The \emph{flatten} operation, flattens a list of list, while keeping the double-underlined element:
        \[ \mu\left(\left[[a,\ \underline{b},\ c],\ \underline{[\underline{d},\ e,\ f]},\ [g,\ \underline{h}]\right]\right) = [a,\ b,\ c,\ \underline{d},\ e,\ f,\ g,\ h]\]
        \item The \emph{extract} operation extracts the underlined element, e.g.: $\counit([a,\ \underline{b},\ c]) = b$. 
        \item The \emph{extend} operation generates a series of new lists, each being a copy of the original list,
              but with a different, consecutive element underlined in each of the copies. Finally, it underlines the copy
              that is exactly equal to the input list (including the underlined element): 
              \[\delta([a,\ \underline{b},\ c ]) = \left[ [\underline{a}, b, c], \underline{[a, \underline{b}, c]}, [a, b, \underline{c}] \right]\]
    \end{enumerate}
\end{definition}

Before we discuss $\overline{L}$-definable transducers, let us formulate a lemma about $\overline{L}$-algebras.
The lemma is based on \cite[Section~8.1]{bojanczyk2015recognisableFull} and shows
that computing products in $\overline{L}$-algebras, boils down to computing two monoid products
(for proof, see Section~\ref{subsec:rational-lenght-preserving} of the appendix):
\begin{lemma}
\label{lem:pointed-alg-prod-two-mon-prod}
    For every $\overline{L}$-algebra $(A, \prod)$, there are two monoids $M_L$ and $M_R$, 
    together with functions $h_L : A \to M_L$, $h_R : A \to M_R$, such that the value of 
    every $A$-product 
    \[\prod([a_1, \ldots, \underline{a_i}, \ldots, a_n])\]
    depends only on: 
    \begin{enumerate}
    \item the $M_L$-product of the prefix (i.e. $h_L(a_1) \cdot \ldots \cdot h_L(a_{i-1})$), 
    \item the $M_R$-product of the suffix (i.e. $h_R(a_{i+1}) \cdot \ldots \cdot h_R(a_{n})$), and
    \item the exact $A$-value of the underlined element ($a_i$).
    \end{enumerate}
    Moreover, if $A$ is finite then both $M_L$ and $M_R$ are finite as well. 
\end{lemma}

We are now ready to discuss $\overline{L}$-definable transductions. Each such transduction 
is, by definition, given by a finite $\bar{L}$-algebra $(A, \prod)$, an input function $h : \Sigma \to A$,
and an output function $\lambda : A \to \Gamma$. A transduction given in this way computes its $i$-th output letter as: 
\[\lambda(\prod([h(a_1), \ldots, \underline{h(a_i)}, \ldots,h(a_n)]))\]
Thanks to \cref{lem:pointed-alg-prod-two-mon-prod}, we know that we know that there are two monoids $M_L$ and $M_R$, 
such that we can compute this value based on the $M_L$-product of the prefix,
$M_R$-product of the suffix and the value of $h(a_i)$. It follows, by \cref{def:monoid-recognisable} of regularity, 
that we can compute the $i$-th letter of the output based on some regular properties of the prefix, 
some regular properties of the suffix, and on the $i$-th input letter. This means that every
$\overline{L}$-definable transduction is also a rational length-preserving function\footnote{
    It might be worth mentioning that there is a slight type mismatch between the types of rational length-preserving
        transductions and $\overline{L}$-definable transductions. The former are of the type $\Sigma^+ \to \Gamma^+$, and the latter 
        of type $\bar{L}(\Sigma) \to \bar{L}(\Gamma)$. We can deal with this mismatch, by observing that
        for the $\bar{L}$-definable transducers, the position of the underlined element does not influence the underlying output word. 
        See Appendix~B.3 for more details. 
}. To prove the other inclusion, we use a similar idea to transform an unambiguous Mealy machine into    $\bar{L}$-algebra. 
(See Appendix~B.3 for a more detailed proof.)


Here is a table that summarizes all classes of $M$-definable transductions, we have seen so far:
\[
    \begin{tabular}{|l|l|l|}
        \hline
        Transduction class        & Machine model                & Functor              \\ \hline
        Sequential left-to-right  & Mealy machines               & $\overrightarrow{L}$ \\
        Sequential right-to-left  & Right-to-left Mealy machines & $\overleftarrow{L}$  \\
        Rational lenght-preserving & Unambigous Mealy machines    & $\overline{L}$       \\ \hline
    \end{tabular}
\]

Observe that all those examples are length-preserving.
This is a consequence of a more general principle, which can be stated using the \emph{shape} of a functor:
\begin{definition}
    Let $1 = \{ \bullet\}$ be a singleton set. For every functor $F$ and every $l \in F X$, we define the \emph{shape}
    of $l$ as the element of $F \, 1$, obtained by replacing every element of $X$ by $\bullet$:
    \[\mathtt{shape}(l) = (F ( x \mapsto \bullet ))\; l\]
\end{definition}
For example, the shapes of both $\overrightarrow{L}$ and $\overleftarrow{L}$ are their lengths,
and the shape $\bar{L}$ is its length and the position of its underlined element.
It is not hard to see that all $M$-definable transductions are shape-preserving: 
\begin{lemma}
\label{lem:definable-preseve-shape}
    For every $M$-definable transduction $F : M \Sigma \to M \Gamma$, and for every $w \in M \Sigma$, it holds that
    $\mathtt{shape}(F(w)) = \mathtt{shape}(w)$.
\end{lemma}

\subsection{Other examples}\label{subsec:other-examples}
\paragraph*{Infinite words}
In this section we extend the monad $C$ from \cref{ex:countable-orders-monads} with 
three different comonad structures $\overrightarrow{C}$, $\overleftarrow{C}$, and $\bar{C}$
(which are analogous to $\overrightarrow{L}$, $\overleftarrow{L}$ and $\bar{L}$ from \cref{subsec:word-to-word-transducer-functors}), 
and we briefly characterize the resulting classes of transducers.

We define $\overrightarrow{C} X \subsetneq C X$, to be the set of all elements of $C X$ that have a maximal element.
To see that this is a monad (with the \emph{singleton} and \emph{flatten} operations inherited from $C$), we observe that
\emph{flatten} preserves the property of having maximal elements.
Next, we define a comonad structure on $\overrightarrow{X}$: \emph{extract} returns the label of the maximal element, 
and \emph{expand} labels each position with its prefix, i.e. the position $i$ of $\delta(l)$ is labelled with:
$\{ x\ |\ x \in l\; \wedge\; x \leq i \}$.
Observe that all such labels contain maximal elements -- 
the maximal element in the label of $i$ is $i$ itself.
As it turns out, the class of $\overrightarrow{C}$-definable transductions admits a logical characterization:
One can show\footnote{The proof follows from the fact that $C$-recognisability is equivalent to \textsc{mso}-definability.
However, since \textsc{mso}-transductions fall slightly out of scope for this paper, we are not going
to give a precise proof. For other transducer classes in this section, we use a similar \textsc{mso}-recognisability argument.}
that it is equivalent to the class of transductions, that preserve the underlying order (see \cref{lem:definable-preseve-shape}),
and compute the new label for each position $i$ based on \textsc{mso}-formulas that only see the positions $\leq i$.
The definition of $\overleftarrow{C}$ and the characterization of $\overleftarrow{C}$-definable transduction are analogous.

Next, we define $\bar{C} X$ to be the set of all elements of $C X$ where exactly one position is underlined.
The monad structure of $\bar{C} X$ is a generalization of the monad structure of $\bar{L}$: 
\emph{singleton} returns a single underlined element, and 
\emph{flatten} flattens the input and underlines the doubly underlined position.
Similarly, the comonad structure of $\bar{C} X$ generalizes the comonad structure of $\bar{L}$: 
\emph{extract} returns the label of the underlined element, and \emph{expand}
labels every position $i$ of its input with a copy of the input where $i$ is the underlined position,
and underlines the copy that corresponds to the underlined position of the input.
The class of $\bar{C}$-definable transduction also admits 
a logical characterization: It is equivalent to the class of transductions that 
preserve the underlying order and the position of the underline,
and compute the output labels based on \textsc{mso}-formulas that see
the entire input.

\paragraph*{Terms}
Finally we define $\bar{T}_{\mathcal{S}}$ to be \emph{pointed} version of the term functor from \cref{ex:term-monad}, 
and we equip it with structures of a monad and a comonad. 
The construction is analogous\footnote{
    It seems that the \emph{pointing} construction is a general way of equipping a monad 
    with a comonad structure. 
} to  $\bar{L}$ and $\bar{C}$.
We define $\bar{T}_\mathcal{S} X$, to be the set of trees (from $\bar{T_\mathcal{S}} X$) with
exactly one underlined leaf. The monadic and comonadic operations are defined analogously 
as for $\bar{L}$ and $\bar{C}$.

As it turns out, the class of $\bar{T_\mathcal{S}}$ definable transductions, also admits a logical characterization:
it is equivalent to the class of tree-to-tree transductions that only modify the labels of the leaves
(this is a consequence of \cref{lem:definable-preseve-shape}) and calculate the 
output label for each leaf based on \textsc{mso}-formulas that have access to the entire input tree.

\section{Compositions of recognizable transducers}\label{sec:compositions-of-recognisable-transductions}
So far we have introduced, and presented a few examples of $M$-recognizable transductions. 
In this section we are going to prove \cref{thm:compose} which states that $M$-recognizable transductions 
are closed under compositions, i.e. if $f : M\,\Sigma \to M\,\Gamma$ and $g : M\,\Gamma \to M \,\Delta$ both
are $M$-recognizable transductions, then so is their composition $g \circ f : M\,\Sigma \to M\,\Delta$.
In addition to the axioms of a functor, monad, and comonad, we have seen so far, 
the proof of the theorem requires some additional \emph{coherence axioms} which relate the monadic and the comonadic structures of $M$.
Here are three examples of such axioms respectively called \emph{flatten-extract}, \emph{singleton-expand}, and \emph{singleton-extract}\footnote{
    To the best of our knowledge, this axiom has not appeared so far in the literature.
}: 
\[\begin{tikzcd}
	{M\, M \, X} && {M \, X} & X && {M X} & X && {M X} \\
	\\
	{M \, X} && X & {M X} && {M M X} &&& X
	\arrow["{\mu_X }", from=1-1, to=1-3]
	\arrow["{\varepsilon_X }", from=1-3, to=3-3]
	\arrow["{\varepsilon_{M X}}"', from=1-1, to=3-1]
	\arrow["{\varepsilon_X}"', from=3-1, to=3-3]
	\arrow["{\eta_X}", from=1-4, to=1-6]
	\arrow["{\eta_X}"', from=1-4, to=3-4]
	\arrow["{M \eta_X}"', from=3-4, to=3-6]
	\arrow["{\delta_X}", from=1-6, to=3-6]
	\arrow["{\eta_X}", from=1-7, to=1-9]
	\arrow["{\varepsilon_X}", from=1-9, to=3-9]
	\arrow["id"', from=1-7, to=3-9]
\end{tikzcd}\]
The other axioms postulate the existence of an additional structure on $M$. 

\subsection{The \texttt{put}-structure and its axioms}
For a comonad $M$, let us consider the following operation\footnote{\label{footnote:lens}The operation $\putf$ has 
already been studied in the context of functional programming. In this context, $M$ does not need to be a (full) comonad, 
it is, however, required to implement the operation $\mathtt{get} : M X \to X$, which in our case is equal to $\counit$. 
In this context, the pair ($\mathtt{get}, \putf$) is usually referred to as a $\emph{lens}$. 
See \cite[Section~1]{fischer2015clear} for details. See \cite[Section~3.1]{foster2007combinators} for the original reference.}:
\[ \putf_A\ : \  M \, A \times A \to M A \]
The intuition behind $\putf$ is that it replaces the \emph{focused element} of $M A$ with the 
given element from $A$. The intuition behind the \emph{focused element} is that this is the element 
that is going to be returned by the \emph{extract} operation. For example, the focused element in $\bar{L}$ is 
the underlined element, and in $\overrightarrow{L}$ it is the last element of the list. Here are two examples of $\putf$:
\begin{center} 
\begin{tabular}{cc}
        $\underbrace{\putf([1, \underline{2}, 3, 4], 7) = [1, 7, 3, 4]}_{\bar{L}}$ & 
        $\underbrace{\putf([1, 2, 3], 5) = [1, 2, 5]}_{\overrightarrow{L}}$
    \end{tabular}
\end{center}

The goal of this subsection is to formalize this intuition in terms of axioms.
First, we assume that $\putf$ is a natural transformation (see \cref{subsec:put-naturality} for details). 
Next, we assume the following axioms that relate $\putf$ and $\counit$. 
They are called \emph{get-put}, \emph{put-get}, and \emph{put-put}\footnote{The axioms and their names come from the \emph{lens}-related research.
See \cite[Section~1]{fischer2015clear}.}:
\[\begin{tikzcd}
	{M \, A} && {M\, A \times A} & {M A \times A} && {M A } & {(M A \times A) \times A } && {M A  \times A} \\
	\\
	&& {M\, A} &&& A & {M A \times A } && {M A }
	\arrow["id"{description}, from=1-1, to=3-3]
	\arrow["{\langle id,\, \varepsilon_A \rangle }"{description}, from=1-1, to=1-3]
	\arrow["{\texttt{put}_A}"{description}, from=1-3, to=3-3]
	\arrow["{\pi_2 }"{description}, from=1-4, to=3-6]
	\arrow["{\varepsilon_A}"{description}, from=1-6, to=3-6]
	\arrow["{\mathtt{put}_A }"{description}, from=1-4, to=1-6]
	\arrow["{\mathtt{put}_A}"{description}, from=1-9, to=3-9]
	\arrow["{\mathtt{put}_A}"{description}, from=3-7, to=3-9]
	\arrow["{\pi_2 \times \mathtt{id}}"{description}, from=1-7, to=3-7]
	\arrow["{\mathtt{put}_A \times id}"{description}, from=1-7, to=1-9]
\end{tikzcd}\]
Here the functions $\langle f, g \rangle$, $f \times g$, $\pi_2$ are defined as follows: 
\[
    \begin{tabular}{ccc}
        $\langle f, g \rangle(x) = (f (x),\ g (x))$ & $(f \times g)(x_1,\ x_2) = (f(x_1),\ g(x_2))$ & $\pi_2(x, y) = y$
    \end{tabular} 
\]
\noindent
The following axioms called \emph{put-associativity} and \emph{singleton-put} relate $\putf$ with the structure of a monad\footnote{To the best of our knowledge this axiom has not appeared previously in the literature.}:
\[\begin{tikzcd}
	{M M A \times M A \times A} && {M M A \times M A } & {M M A } & {A \times A} && {MA \times A} \\
	\\
	{M M A \times A } && {M A \times A} & {M A } & A && {M A}
	\arrow["{\mathtt{put}_A}", from=3-3, to=3-4]
	\arrow["{\mu_A \times id}", from=3-1, to=3-3]
	\arrow["{\mathtt{put}_{M A} \times id}"{description}, from=1-1, to=3-1]
	\arrow["{id \times \mathtt{put}_A}", from=1-1, to=1-3]
	\arrow["{\mathtt{put}_{M A}}", from=1-3, to=1-4]
	\arrow["{\mu_A}"{description}, from=1-4, to=3-4]
	\arrow["{\pi_2}", from=1-5, to=3-5]
	\arrow["{\eta_A}", from=3-5, to=3-7]
	\arrow["{\eta_A \times id }"', from=1-5, to=1-7]
	\arrow["{\mathtt{put}_A}"', from=1-7, to=3-7]
\end{tikzcd}\]
The final axiom relates all the structures studied in this paper: monad, comonad, and $\mathtt{put}$.
Before we present it, we need to define the \emph{strength of a functor}\footnote{
    This definition of $\strength$ is specific to \Set.
    We briefly discuss other categories in \cref{subsec:ccc}.
}.
\begin{definition}\label{def:strength}
    For every functor $F$, we define $\strength_{(A, B)} : A \times M B \to M (A \times B)$:
    \[ \strength(a, l) = F ( x \mapsto  (a, x) )\ l  \]
    Intuitively, the function $\strength_{(A, B)}$ equips each element $B$ under the functor $F$
    with a copy of $a \in A$.
    Here is an example, for $F$ equal to the list functor:
    \[ \strength(c, [a, b, a, b]) = [(c, a), (c, b), (c, a), (c, b)]\]
\end{definition}

\noindent
We are now ready to present the final coherence axiom, called \emph{flatten-expand}\footnote{To the best of our knowledge the axiom has not appeared before in the literature.}:
\[\begin{tikzcd}
	& {M M M A} && {M M M A} \\
	\\
	{M M A} &&&& {M M A } \\
	&& {M A}
	\arrow["{\mu_A}"{description}, from=3-1, to=4-3]
	\arrow["{\delta_A}"{description}, from=4-3, to=3-5]
	\arrow["{\mu_{M A}}"{description}, from=1-4, to=3-5]
	\arrow["{\delta_{M A}}"{description}, from=3-1, to=1-2]
	\arrow["{M \mathtt{work}}", from=1-2, to=1-4]
\end{tikzcd}\]
where the function $\mathtt{work}$ is defined as the following composition:
\[ M M A \transform{\langle \id, \counit \rangle } M M A \times M A \transform{\id \times \delta} M M A \times M M A \transform{\mathtt{strength}} M (M M A \times M A)
\transform{M \putf} M M M A \transform { M \mu} M A \]
Let us now briefly present the intuition behind the flatten-expand axiom. 
The starting point $ M M X$ represents a structure partitioned into substructures
(e.g. a list partitioned into sublists), which we would like to expand using $\delta$.
The bottom path of the diagram represents the straightforward approach: First, 
it flattens the input using $\mu_A$ (forgetting about the substructure partitions),
and then it applies the $\delta_A$ function.
The flatten-expand axiom asserts that this can be done in a way that 
respects the initial partitions. This way is represented by the 
top path of the diagram:
First, it applies the $\delta_{M A}$ function to expand the top structure; 
then it applies the $\mathtt{work}$ function
independently to each of the substructures using $M\, \mathtt{work}$
(this can also be seen as a concurrent computation); 
and finally it aggregates the results of $\mathtt{work}$ using $\mu_{M A}$.
(See \cref{subsec:flatten-expand-ex} for a step-by-step example.)
It might also be worth mentioning that the flatten-expand axiom has an alternative
formulation in terms of a \emph{bialgebra} (see \cref{subsec:flatten-expand-as-bialgebras}). 

Finally, let us mention that it is not hard to all the examples of monad/comonad functors from \cref{sec:comonads-and-transdcuers}
with the natural $\putf$ operation,  and show that they satisfy all the axioms we have introduced.


\subsection{Contexts}\label{subsec:ctx}
In this section we use the $\putf$ structure to introduce \emph{contexts}
for Eilenberg-Moore algebras. This concept plays an important role 
in the proof of \cref{thm:compose}. Additionally, we would like to 
highlight the potential of contexts as an independently interesting 
tool for studying Eilenberg-Moore algebras -- this point is 
illustrated by \cref{def:M-group}. In this subsection,
we only assume that $M$ is a monad equipped with the $\putf$ structure -- 
it does not depend on the comonad structure of $M$. 
\begin{definition}
    Let $(A, \prod)$ be an $M$-algebra. 
    For every element $l \in M A$, we define its \emph{context} to be the following 
    function $\ctx_l : A \to A$:
    \[ \ctx_l(x) = \prod(\putf(l, x)) \]
    In other words, the context of $l$ takes an element $x \in A$, replaces the focused element of $l$ with $x$,
    and calculates the product of the resulting $M A$. 
\end{definition}
Let us now discuss some properties of the contexts.
Thanks to the put-put axiom one can show that the context of $l$ does not depend on its focused element
(the formal proof is verified in Coq as \texttt{ctxPutInvariant}, see \cref{subsec:coq-contexts}).
\begin{lemma}\label{lem:ctx-put}
    For every $l \in MA$, and every $a \in A$, the context of $l$ is equal to the context of $\putf(l, a)$, i.e. 
    $\ctx_l = \ctx_{\putf(l, a)}$. 
\end{lemma}
Similarly, using the singleton-put axiom, one can show that the context of every singleton is the identity function
(the formal proof is verified in Coq as \texttt{ctxUnitId}, see \cref{subsec:coq-contexts}).
\begin{lemma} \label{lem:ctx-id}
    For every $a \in A$, it holds that $\ctx_{(\eta_A\, a)} = \id$. 
\end{lemma}
Now, let us consider the set of all contexts:
\begin{definition}\label{def:contexts}
For every $M$-algebra $(A, \alpha)$, we define the set $C_A \subseteq (A \to A)$ to be the set of all possible contexts, 
i.e. $C_A = \{ \ctx_l \ | \ l \in M A\}$.
\end{definition}
The important property of $C_A$ is that it is closed under compositions:
\begin{lemma}\label{lem:contexts-compose}
    For every $f, g \in C_A$, it holds that $f \circ g \in C_A$. 
\end{lemma}

It follows that $C_A$ is a transformation monoid of $A$. 
It follows that the mapping $A \mapsto C_A$ allows us to transform
an arbitrary $M$-algebra into a monoid. This is why, we believe that 
contexts, are an interesting tool for studying Eilenberg-Moore algebras.
To illustrate this point let us show how to generalize the definition of a \emph{group}:
\begin{definition}
\label{def:M-group}
    We say that an $M$-algebra $A$ is an \emph{$M$-group} if $C_A$ is a group 
    (i.e. for every function $f \in C_A$, its inverse $f^{-1}$ also belongs to $C_A$).  
\end{definition}
In order to validate this definition, let us show that for $M = \overrightarrow{L}$
the definition of $M$-group coincides with the usual definition of a group\footnote{
    If we only consider the monad structure then $\overrightarrow{L}$ is equal to $X^{+}$ (i.e. the monad of non-empty lists).
    However, we point out that we consider $\overrightarrow{L}$, because 
    the definition of $\putf$ depends on whether we consider $\overrightarrow{L}$ 
    or $\overleftarrow{L}$. (The proof for $\overleftarrow{L}$ is, however, analogous). 
} (remember that $\overrightarrow{L}$-algebras are semigroups). For this, let us fix
an $\overrightarrow{L}$-algebra $S$ (i.e. a semigroup). Observe now that, by definition,  
the context of an element $[s_1, \ldots, s_n]$ is equal to the following function:
\[ x \mapsto s_1 \cdot \ldots \cdot s_{n-1} \cdot x \]
It follows that every element of $C_S$ is of the form $x \mapsto s x$, 
where $s$ is an element of $S^1$, where $S^1$ is the smallest monoid that contains $S$, i.e.:
\[
    S^1 = \begin{cases}
        S & \textrm{if $S$ already contains and identity element}\\
        S + \{1\} & \textrm{otherwise}
    \end{cases}
\]
Here $1$ denotes the \emph{formal identity element} whose operations are defined as $1 \cdot x = x \cdot 1 = x$.
Moreover $C_S$ is isomorphic to $S^1$:
\[ (x \mapsto s_1 \cdot x) \circ (x \mapsto s_2 \cdot x ) = (x \mapsto s_1 \cdot s_2 \cdot x) \]
This finishes the proof, as it is not hard to see that $S$ is a group if and only if $S^1$ is a group. 

\subsection{Composition theorem}\label{subsec:composition-theorem}
We are now ready to formulate and prove the main theorem of this article: 
\begin{theorem}
\label{thm:compose}
    Let $M$ be a functor that is both a monad and a comonad, 
    for which there exists a $\putf : M A \times A \to M A$, that
    satisfies the axioms mentioned in this section, i.e.:
    flatten-extract, get-put, put-get, put-associativity, and flatten-expand. 
    Then, the class of $M$-definable transductions is closed under compositions. 
\end{theorem}
The reminding part of this section is dedicated to proving \cref{thm:compose}. After unfolding 
the definitions, this boils down to showing that for each pair of $M$-algebras $(S_1, \prod_1)$, $(S_2, \prod_2)$,
and for every $h_1 : \Sigma \to S_1$, $h_2 : \Gamma \to S_2$, $\lambda_1: S_1 \to \Gamma$, 
$\lambda_2 : S_2 \to \Delta$, there exists an $M$-algebra $(S_3, \prod_3)$ and functions $h_3 : \Sigma \to S_3$,
$\lambda_3 : S_3 \to \Delta$, that make the following diagram commute:
\[\begin{tikzcd}
	{M \Sigma} & {M S_1} & {M MS_1 } & {M S_1} & {M \Gamma} \\
	& {M S_3} &&& {M S_2} \\
	&& {M M S_3} && {M M S_2} \\
	&&& {M S_3} & {M S_2} \\
	&&&& {M \Delta}
	\arrow["{M h_1}", from=1-1, to=1-2]
	\arrow["\delta", from=1-2, to=1-3]
	\arrow["{M \prod_1}", from=1-3, to=1-4]
	\arrow["{M \lambda_1}", from=1-4, to=1-5]
	\arrow["{M h_2}", from=1-5, to=2-5]
	\arrow["\delta", from=2-5, to=3-5]
	\arrow["{M \prod_2}", from=3-5, to=4-5]
	\arrow["{M \lambda_2}", from=4-5, to=5-5]
	\arrow["{M h_3}"', dashed, from=1-1, to=2-2]
	\arrow["\delta"', dashed, from=2-2, to=3-3]
	\arrow["{M \prod_3}"', dashed, from=3-3, to=4-4]
	\arrow["{M \lambda_3}"', dashed, from=4-4, to=5-5]
\end{tikzcd}\]
As our $S_3$ we are going to use the following set:
\[ S_3 = S_1 \times (S_1^{S_1} \to S_2 ) \]
(Note that $S_1^{S_1}$ is a notation for $S_1 \to S_1$ -- we mix the arrow notation and the exponent notation for visual clarity.)
Because of its similarities with the \emph{wreath product} for semigroups, we call our construction for $(S_3, \alpha_3)$ as the \emph{generalized wreath product}, or \emph{$M$-wreath product} 
of $(S_1, \alpha_1)$ and $(S_2, \alpha_2)$. (In the appendix we give a definition of the classical wreath product and compare it with the generalized wreath product -- reading this
part of the appendix could make it easier to understand the construction presented below.)
Before we define $\alpha_3$, $h_3$ and $\lambda_3$, we describe what we would like the composition\footnote{
By \cite[Lemma~4.7]{bojanczyk2020languages}, this composition could be used to define the product operation on $S_3$.
However, since the lemma has extra assumptions, we show this function only for intuition.} $M \Sigma \transform{M h_3} M S_3 \transform{\alpha_3} S_3$ to do --
this way we can present some intuitions behind $S_3$. We start with a $w \in M \Sigma$
and we would like to produce a pair $S_1 \times (S_1^{S_1} \to S_2)$. The first component (i.e. $S_1$) is simply defined as the 
$S_1$-product of the input: 
\[M \Sigma \transform{M h_1} M S_1 \transform{\alpha_1} S_1\]
The interesting part is the second component (i.e. $S_1^{S_1} \to S_2$), which represents 
the $S_2$ product of the input. In order to compute it, we first apply the first $M$-transduction (i.e. ($S_1$, $h_1$, $\lambda_1$)), 
and then we compute the $S_2$-product of the result. This means that the $S_2$ product depends on the $S_1$-context 
in which we evaluate the input, so we provide it as the $S_1^{S_1}$ argument
(intuitively we are only interested in the functions from $C_{S_1}$, as defined in \cref{subsec:ctx}, 
but the definition makes formal sense for all functions $S_1^{S_1}$).
Here is how to compute the $S_2$-value based on the input word $w \in M \Sigma$ and the context $c \in S_1^{S_1}$:
We start by computing the $S_1$-products of the views (while, for now, ignoring the context $c$):
\[ S_1^{S_1} \times M \Sigma \transform{\id \times M h_1} S_1^{S_1} \times M S_1 \transform{\id \times \delta} S_1^{S_1} \times M M S_1 \transform{\id \times M \alpha_1} S_1^{S_1} \times M S_1\]
Next, we apply the context $c$ to each of the prefix products, and compute its $\Gamma$-value: 
\[S_1^{S_1} \times M S_1  \transform{\strength} M (S_1^{S_1} \times S_1) \transform {M \mathtt{app}} M S_1 \transform{M \lambda_1} M \Gamma\]
Finally, we compute the $S_2$-product of the result: 
\[ M \Gamma \transform{ M h_2} M S_2 \transform{\alpha_2} S_2\]

We are now ready to define $\lambda_3$, $h_3$, and $\alpha_3$.
In order to compute $\lambda_3$, we compute the $S_2$-value in the empty context (represented by $\id \in S_1^{S_1}$), and then apply $\lambda_2$: 
\[ \lambda_3((v_1, v_2)) = \lambda_2(v_2(\id)) \]
The function $h_3 : \Sigma \to S_1 \times (S_1^{S_1} \to S_2)$ is defined as follows: In order to compute the $S_1$ value we simply apply 
$h_1$ to the input letter, and in order to compute the $S_2$ value given the context $c \in S_1^{S_1}$, we apply $h_1$, $c$, $\lambda_1$ and $h_2$:
\[ h_3(a) = \left( h_1(a),\ \  c \mapsto h_2\left(\lambda_1(c(h_1(a)))\right)\  \right) \]
Finally, we define the product operation:
\[\alpha_3 : M ( S_1 \times (S_1^{S_1} \to S_2) ) \to  (S_1 \times (S_1^{S_1} \to S_2)) \]
We define $\alpha_3$ using two auxiliary functions $f_1$ and $f_2$:
\[
    \alpha_3(l) = \left(\, f_1(l),\  c \; \mapsto \;  f_2(c, l) \, \right)
\]
The first function $f_1 : M (S_1 \times (S_1^{S_1} \to S_2)) \to S_1$ computes the product of the $S_1$-values:
\[M ( S_1 \times (S_1^{S_1} \to S_2) ) \transform{M \pi_1} M S_1 \transform{\alpha_1} S_1\] 
The second function $f_2 : S_1^{S_1} \times M (S_1 \times (S_1^{S_1} \to S_2)) \to S_2$ is more complicated. We start by computing for each element, 
its view on the $S_1$-values while keeping its $(S_1^{S_1} \to S_2)$-value:
\[ S_1^{S_1} \times M (S_1 \times (S_1^{S_1} \to S_2)) \transform{\id \times \delta } S_1^{S_1} \times M M (S_1 \times (S_1^{S_1} \to S_2)) \transform{\id \times \langle M \pi_1, \pi_2 \circ \counit \rangle} 
S_1^{S_1} \times M ( (M S_1) \times (S_1^{S_1} \to S_2) )\]
Then we compute the context for each of those $S_1$-views:
\[S_1^{S_1} \times M ( (M S_1) \times (S_1^{S_1} \to S_2) )  \transform{\id \times M (\ctx \times \id)} S_1^{S_1} \times M ( S_1^{S_1} \times (S_1^{S_1} \to S_2) )  \]
Next, we compose the initial context with each of the intermediate contexts:
\[ S_1^{S_1} \times M ( S_1^{S_1} \times (S_1^{S_1} \to S_2) ) \transform{\strength} M (S_1^{S_1} \times S_1^{S_1} \times (S_1^{S_1} \to S_2)) \transform {M ((\circ) \times \id)} M (S_1^{S_1} \times (S_1^{S_1} \to S_2))\]
Now, in each position, we apply the function to the argument:
\[M (S_1^{S_1} \times (S_1^{S_1} \to S_2)) \transform{M ((x,f) \mapsto f(x))} M S_2 \]
Finally, we compute the product of the $S_2$ values: 
\[ M S_2 \transform{\alpha_2} S_2\]
This finishes the construction of $(S_3, \alpha_3)$. Now we need to show that it is indeed an $M$-algebra:
\begin{lemma} \label{lem:wreath-product-assoc}
The generalized wreath product $(S_3, \alpha_3)$, as defined above, is a valid $M$-algebra, i.e. for every $l \in M M S_3$, and every $x \in S_3$
it satisfies the following axioms (see \cref{def:em-algabra}):
\[\begin{tabular}{ccc} 
    $\alpha_3 ( \mu (l)) = \alpha_3((M \alpha_3)(l))$ & and & $\alpha_3(\eta(x)) = x$
\end{tabular}
\]
\end{lemma}
The proof of \cref{lem:ctx-id} is quite complex -- the main reason for this is that the definition of $(S, \alpha_3)$ is rather involved.
In contrast, the idea behind the proof is straightforward: we unfold all definition and perform equational reasoning using the axioms.
For this reason, we decided to formalize the proof in the Coq theorem prover --
it can be found as theorems \texttt{S3Associative} and \texttt{S3UnitInvariant} in the attached Coq file,
see \cref{subsec:composition-theorem-coq}.
Finally, we show that the $M$-transduction $(S_3, h_3, \lambda_3)$ computes the required compositions:
\begin{lemma}\label{lem:wreath-correct}
    The $M$-transduction $(S_3, h_3, \lambda_3)$ is equivalent to the composition of $M$-transductions $(S_1, h_1, \lambda_1)$ and $(S_2, h_2, \lambda_2)$. 
\end{lemma}
Similarly as for \cref{lem:wreath-product-assoc}, we prove \cref{lem:wreath-correct} by unfolding the definitions and applying the equational reasoning.
It is called \texttt{compositionCorrect} in the formalization, see \cref{subsec:composition-theorem-coq}.
The proof of \cref{lem:wreath-correct} finishes the proof of \cref{thm:compose}.

\section{Further work}\label{sec:futher-work}
\begin{description}
\item[1. Shape-modifying transductions.] \label{it:shape-modifying} Many important classes of transduction can modify the 
shape of their inputs. Examples of such classes for word-to-word transductions include
\emph{regular transductions} (defined, for example, by \emph{two-way transducers}, or \emph{\textsc{mso}-transductions)}
or \emph{polyregular transductions} (defined, for example, by \emph{for programs} or \emph{\textsc{mso}-interpretations}~\cite{bojanczyk2019string}).
Extending the definitions of $M$-definable transductions to capture those classes is, in our opinion, an interesting 
research direction. As a first step towards this goal, let us propose the following relaxation of $M$-definable transduction. 
The output function $\lambda$ is of type $A \to M \Gamma$ (rather than $A \to \Gamma$), and the transduction is defined as follows:
\[ M \Gamma \transform{M h} M A \transform{\delta} M M A \transform{M \alpha} M A \transform{M \lambda} M M A \transform{\mu} M A\]
For example, for $M = \overrightarrow{L}$ this new class corresponds where the transitions are allowed
to output more than one letter (but have to output at least one letter).

\item[2. Aperiodicity.] We say that a semigroup $S$ is \emph{aperiodic}, if there is no monomorphism $G \to S$,
where $G$ is a non-trivial group. This is an importation notion in the theory of regular languages and transducers.
For example, it very often coincides with \emph{first-order definability} (e.g. \cite{schutzenberger1965finite}).
Thanks to \cref{def:M-group}, we can extend the definition of aperiodicity to arbitrary $M$-algebras 
(for $M$s that are monads and comonads, and are equipped with the $\putf$). Studying this new notion 
of \emph{generalized aperiodicity} could be an interesting research direction.

\item [3. Krohn-Rhodes decompositons.] The Krohn-Rhodes decomposition theorem (\cite[Theorem~(ii)]{krohn1965algebraic}) shows how to present
every semigroup with wreath products of groups and a $3$-element monoid called 
\emph{flip-flop monoid}. Its original proof starts by decomposing Mealy machines, 
and then it shows how to decompose monoids. 
Since in our paper, we generalize the definitions of a group, a wreath product, and a Mealy machine, 
we believe that there is potential for generalizing the original Krohn-Rhodes theorem to $M$-algebras.

\item [4. Other categories.] A natural follow-up of this paper would be generalizing it from 
the category \Set \ to arbitrary Cartesian closed categories. For now the biggest obstacle to such a 
generalization seems to be the $\strength$ function from \cref{def:strength}. See \cref{subsec:ccc} for more details. 
\end{description}
\clearpage
\bibliography{moconads}
\clearpage 
\appendix

\section{Omitted details from \cref{sec:monads-and-languages}}
\subsection{Monad axioms} \label{subsec:monad-axioms}
Let us present the omitted axioms of a monad: 
First, both $\eta$ and $\mu$ should be natural. In this particular case, this means
 that the following two diagrams should commute\footnote{
    We hope that the notation of commutative diagrams is self-explanatory. 
    See \cite[Section~1.1.3]{abramsky2011introduction} for a formal description.
 }: for every function $f : X \to Y$ (for the general
 definition of naturality, see \cite[Definition~1.4.1]{abramsky2011introduction}
 or \cite[Section~10]{milewski2018category}):

\[\begin{tikzcd}
	{M \, M \, X} &&& {M\, X} && X &&& {M\, X} \\
	\\
	{M \, M \, Y} &&& {M \, Y} && Y &&& {M \, Y}
	\arrow["{\mu_X}", from=1-1, to=1-4]
	\arrow["{M (M\,  f)}"', from=1-1, to=3-1]
	\arrow["{M \, f}", from=1-4, to=3-4]
	\arrow["{\mu_Y}", from=3-1, to=3-4]
	\arrow["f", from=1-6, to=3-6]
	\arrow["{M\, f}", from=1-9, to=3-9]
	\arrow["{\eta_X}", from=1-6, to=1-9]
	\arrow["{\eta_Y}", from=3-6, to=3-9]
\end{tikzcd}\]

In addition to being natural, $\eta$ and $\mu$ should make the following diagrams commute:
\[\begin{tikzcd}
	{M\, M \, M \, X} &&& {M\, M\, X} && {M\, X} && {M \, M \, X} \\
	\\
	{M \, M \, X} &&& {M \, X} && {M\, M\, X} && {M\, X}
	\arrow["{\mu_{MX}}", from=1-1, to=1-4]
	\arrow["{\mu_X}", from=1-4, to=3-4]
	\arrow["{\mu_X}", from=3-1, to=3-4]
	\arrow["{M \mu_X}"', from=1-1, to=3-1]
	\arrow["{\eta_{M X}}"', from=1-6, to=1-8]
	\arrow["{\mu _X}"', from=1-8, to=3-8]
	\arrow["{\textrm{id}}"{description}, from=1-6, to=3-8]
	\arrow["{M\, \eta_X}"', shift left, from=1-6, to=3-6]
	\arrow["{\mu_X}", from=3-6, to=3-8]
\end{tikzcd}\]

\section{Omitted details from \cref{sec:comonads-and-transdcuers}}
\subsection{Comonad axioms}\label{subsec:comonad-axioms}
Let us present the omitted axioms of a comonad.
First, both $\counit$ and $\delta$ have to be natural, i.e. for every $f : X \to Y$ they have to satisfy the following 
        commutative equations:
    \[\begin{tikzcd}
	{M\, X} &&& {M\, M \, X} && {M\, X} &&& X \\
	\\
	{M \, Y} &&& {M\, M\, Y} && {M \, Y} &&& Y
	\arrow["Mf", from=1-1, to=3-1]
	\arrow["{M (M f)}", from=1-4, to=3-4]
	\arrow["{\delta_X}", from=1-1, to=1-4]
	\arrow["{\delta_Y}", from=3-1, to=3-4]
	\arrow["{M\, f}", from=1-6, to=3-6]
	\arrow["{\varepsilon_X}", from=1-6, to=1-9]
	\arrow["{\varepsilon_Y}", from=3-6, to=3-9]
	\arrow["f", from=1-9, to=3-9]
    \end{tikzcd}\]

    In addition to being natural, $\delta$ and $\varepsilon$ should make the following diagrams commute:
\[\begin{tikzcd}
	{M\, X} && {M\, M\, X} && {M \, X} && {M \, M \, X} \\
	\\
	{M\, M\, X} && {M \, M \, M\, X} && {M\, M\, X} && {M \, X}
	\arrow["{\delta_X}"', from=1-1, to=3-1]
	\arrow["{\delta_X}", from=1-1, to=1-3]
	\arrow["{\delta_{M X}}", from=1-3, to=3-3]
	\arrow["{M \, \delta_X}"', from=3-1, to=3-3]
	\arrow["{\delta_{X}}", from=1-5, to=1-7]
	\arrow["{\delta_X}", from=1-5, to=3-5]
	\arrow["{M\, \varepsilon_X}", from=3-5, to=3-7]
	\arrow["{\varepsilon_{MX}}", from=1-7, to=3-7]
	\arrow["{\textrm{id}}"{description}, from=1-5, to=3-7]
\end{tikzcd}\]

\subsection{Mealy machines}\label{subsec:mealy}
Mealy machines are one of the most basic, and very well-studied models of transducers. 
They were introduced by \cite{mealy1955method}. In this section, we give a full definition of Mealy machines,
and show that they are equivalent to $\overrightarrow{L}$-definable transductions. 
\begin{definition}
    Let $\Sigma$ and $\Gamma$ be finite alphabets. A \emph{Mealy machine} of type $\Sigma^+ \to\Gamma^+$ consists~of: 
    \begin{enumerate}
        \item a finite set of states $Q$;
        \item an initial state $q_0 \in Q$; and
        \item a transition function: 
            \[ \delta : \underbrace{Q}_{\textrm{current state}} \times \underbrace{\Sigma}_{\textrm{input letter}} \to
                        \underbrace{Q}_{\textrm{new state}} \times \underbrace{\Gamma}_{\textrm{output letter}} \]
    \end{enumerate}
    (Observe that contrary to a deterministic Mealy machine does not have accepting states.)

    A Mealy machine defines the following function\footnote{
    Usually the type of a Mealy machine's function is defined as $\Sigma^* \to \Gamma^*$. 
    However, this does not make much difference. This is because Mealy machines are a length preserving model, 
    so for the empty input they always return the empty output, and for a non-empty input they always 
    return a non-empty output. It follows that the function $\Sigma^* \to \Gamma^*$ is uniquely defined 
    by the function $\Sigma^+ \to \Gamma^+$. 
    } $\Sigma^+ \rightarrow \Gamma^+$.
    It starts in the initial state, and processes the input word letter by letter. 
    For each symbol $w_i$, the machine transitions to a new state $q'$ and outputs a symbol $x \in \Gamma$, 
    as directed by the transition function $\delta(q, w_i) = (q', x)$. When the machine has processed it entire input, 
    the output word is obtained as the sequence of all letters from $\Gamma$ calculated by the machine. 
\end{definition} 
For example, here is a Mealy machine of type $\{a, b\}^+ \to \{c, d\}^+$ that calculates the function
``Change the first $a$ to $c$, and all other letters to $d$'':
\smallpicc{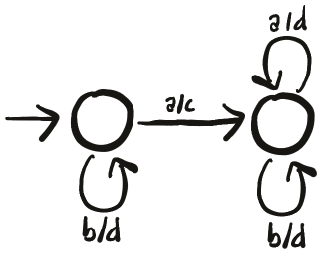}

\begin{lemma}\label{lem:mealy-overrightarrow-eq}
    The class of transductions recognized by Mealy machines is equivalent to $\overrightarrow{L}$-recognizable 
    transductions.
\end{lemma}
\begin{proof}
    $\subseteq$: As explained in  \cref{subsec:word-to-word-transducer-functors},
    every $\overrightarrow{L}$-recognizable transduction is given by 
    a semigroup $S$, an input function $h : \Sigma \to S$ and an output function $S \to \Gamma$. 
    Such a transduction computes the following function:
    \[ \begin{tabular}{ccc}
        $\underbrace{a_1\ a_2\ \ldots\ a_n}_{\Sigma^+} $ & $\mapsto$ & $\underbrace{\lambda\left( h(a_1)\right),\  \lambda\left(h(a_1) \cdot h(a_2)\right),\ \ldots,\  \lambda\left(h(a_1) \cdot \ldots \cdot h(a_n)\right)}_{\Gamma^+}$
    \end{tabular}
    \]  
    Let us show how to translate such a transduction into a Mealy machine: We start by extending the semigroup $S$ to 
    a monoid $S^I = S \cup \{1\}$, where $1$ the formal identity element (i.e. $1 \cdot x = x \cdot 1 = x$ for every $x \in S^I$).
    Now we say the set of states of the Mealy machine is equal to $S'$, its initial state is $1$, and its transition function 
    is given by the following formula:
    \[ \delta(s, a) = (s \cdot h(a), \lambda(s\cdot h(a)))\]
    This way the Mealy machine computes the $S$-products of the $h$-values of the input prefixes, 
    and outputs their $\lambda$-values, computing the same function as the original $\overrightarrow{L}$-recognizable transduction.

    $\supseteq$: Now, we are given a Mealy machine of type $\Sigma^+ \to \Gamma^+$ given by $(Q, q_0, \delta)$, 
    and we want to construct $(S, h, \lambda)$, such that the $\overrightarrow{L}$-transduction given by 
    $(S, h, \lambda)$ is equivalent to the initial Mealy machine.
    For this purpose let us define the behaviour of an infix $w \in \Sigma^+$, which 
    is an element from the following set:
    \[ S = \underbrace{Q}_{\substack{
            \textrm{The state in which}\\
            \textrm{the Mealy machine enters}\\
            \textrm{the infix from the left}
    }} \to \underbrace{Q}_{\substack{
        \textrm{The state in which}\\
        \textrm{the Mealy machine exists}\\
        \textrm{the infix from the right} 
    }} \times \underbrace{\Gamma}_{\substack{
        \textrm{The letter that}\\
        \textrm{the Mealy machine outputs}\\
        \textrm{while exiting the infix}
    }} \]
    Observe that those behaviours are compositional, if we know that the behaviours of words 
    $w, v \in \Sigma^+$ are equal respectively to $f_w$ and $f_v$, then we know that the behaviour 
    of $wv$ is equal to the following function:
    \[ f_{wv}(q) = f_v(\pi_1(f_w(q)))\comma\]
    where $\pi_1 : X \times Y \to X$ represents the projection to the first coordinate.
    This gives us a semigroup structure on the set of behaviours (where $f \cdot g$ is defined 
    as $g \circ \pi_1 \circ f$). We take this monoid of behaviours as our semigroup $S$. 
    Then, we define $h: \Sigma \to S$ and $\lambda : S \to \Gamma$ in the following way:
    \[\begin{tabular}{cc}
        $h(a) = q \mapsto \delta(q, a)$ & $\lambda(f) = \pi_2(f(q_0))$
    \end{tabular}
    \]
    This way, the $\overrightarrow{L}$ transduction given by $(S, h, \lambda)$ computes the 
    behaviour of each prefix, and outputs the letter that the machine would output if it entered 
    the prefix in $q_0$. It is not hard to see that this computes the same function as the original Mealy machine. 
\end{proof}

\subsection{Rational length-preserving transductions}\label{subsec:rational-lenght-preserving}
In this section, we present the definition of \emph{rational length-preserving} transductions.
We define them using \emph{unambiguous Mealy machines}\footnote{Although, we could not find 
a reference to this exact model in the literature, we believe that it belongs to the field's
folklore, as it can be seen as a length-preserving version of the functional NFA with output 
(see \cite[Chapter~25]{eilenberg1974automata}, or \cite[Section~13.2]{bojanczyk2018automata}).}
and show that they are equivalent to $\bar{L}$-definable transductions.
\begin{definition}
    Let $\Sigma$ and $\Gamma$ be finite alphabets.
    A \emph{nondeterministic Mealy machine} of type $\Sigma^+ \to \Gamma^+$ consists of:
    \begin{enumerate}
        \item A finite set of states $Q$. 
        \item A subset $I \subseteq Q$ of initial states, and a subset $F \subseteq Q$ of final (i.e. accepting) states. 
        \item A transition relation:
        \[ \delta : \underbrace{Q}_{\textrm{current state}} \times \underbrace{\Sigma}_{\textrm{input letter}} \times
        \underbrace{Q}_{\textrm{new state}} \times \underbrace{\Gamma}_{\textrm{output letter}} \]
    \end{enumerate}
    A \emph{run} of a nondeterministic Mealy machine over an input word $w \in \Sigma^+$ is a sequence of states,
    starting from an initial state $q_0 \in I$, and ending in a final state $q_n \in F$, 
    such that for each symbol $w_i$ of $w$,
    there is a transition in $\delta$ that reads $w_i$ and takes the machine from state $q_{i-1}$ to state $q_i$.
    Observe that each run produces an output word in $\Gamma^+$.
    The machine is called \emph{unambiguous} if, for every input word $w \in \Sigma^+$, there exists \emph{exactly} one run.
    The transduction defined by an unambiguous Mealy machine is a function $\Sigma^+ \to \Gamma^+$ that maps a word $w \in \Sigma^+$ 
    to the output of the machine's only run for $w$. (Observe that the unambiguity of the machine guarantees that for every input, 
    there is exactly one output, despite the machine's nondeterministic transition relation.)
\end{definition}
For example, here is an unambiguous Mealy machine of type $\{a, b\}^+ \to \{a, b\}^+$ that computes the function
``replace the first letter with the last one'':
\custompicc{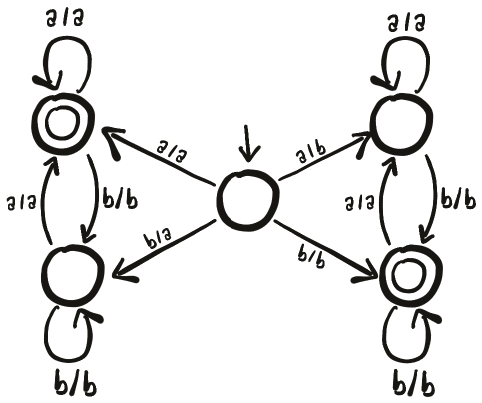}{0.4}

Observe now that there is a slight type mismatch between the types of rational length-preserving
transductions and $\overline{L}$-definable transductions. The former are of the type $\Sigma^+ \to \Gamma^+$, and the latter 
of type $\bar{L}(\Sigma) \to \bar{L}(\Gamma)$. To deal with this mismatch, we notice that
all $\overline{L}$-definable transductions satisfy the following property
(this is an immediate consequence of the definition of $\bar{L}$-definable transductions):
\begin{definition}
    We say that a length-preserving function $f : \bar{L} X \to \bar{L} Y$ is \emph{underline-independent} if:
    \begin{enumerate}
        \item For every $v \in \bar{L}$, the underlying word of $f(v)$ does not depend on the position of the underline in $w$, i.e. 
              for every $w \in X^+$:
              \[  \mathtt{forget}(f( \mathtt{underline}_i(w) )) = \mathtt{forget}(f( \mathtt{underline}_j(w) ))\comma\]
              where $\mathtt{underline}_i$ is the function that underlines the $i$th element of the input, and $\mathtt{forget}$ is 
              the function that casts a pointed list into a normal list (by erasing the underline).
        \item For every $v \in \bar{L}$ the index of the underlined position in $v$ is equal to the index of the underlined position 
              in $f(v)$. 
    \end{enumerate}
\end{definition}
Observe now that the following function is a bijection between length-preserving, underline-independent functions $\bar{L}(\Sigma) \to \bar{L}(\Gamma)$ and
length-preserving functions $\Sigma^+ \to \Gamma^+$:
\[
    \phi(f) = \mathtt{forget} \circ f \circ \mathtt{underline}_1
\]
From now on, we are going to use this bijection implicitly, equating the two types of functions.

Next, let us show that unambiguous Mealy machines are equivalent to $\bar{L}$-definable transductions. 
We start with the proof of \cref{lem:pointed-alg-prod-two-mon-prod}
(the lemma and its proof are based on \cite[Section~8.1]{bojanczyk2015recognisableFull}):
\begin{replemma}{\ref{lem:pointed-alg-prod-two-mon-prod}}
        For every $\bar{L}$-algebra $(A, \prod)$, there are two monoids $M_L$ and $M_R$, 
        together with functions $h_L : A \to M_L$, $h_R : A \to M_R$, such that the value of 
        every $A$-product 
        \[\prod([a_1, \ldots, \underline{a_i}, \ldots, a_n])\]
        depends only on: 
        \begin{enumerate}
        \item the $M_L$-product of the prefix (i.e. $h_L(a_1) \cdot \ldots \cdot h_L(a_{i-1})$), 
        \item the $M_R$-product of the suffix (i.e. $h_R(a_{i+1}) \cdot \ldots \cdot h_R(a_{n})$), and
        \item the exact $A$-value of the underlined element ($a_i$).
        \end{enumerate}
        Moreover, if $A$ is finite then both $M_L$ and $M_R$ are finite as well. 
\end{replemma}
\begin{proof}
      For every element $a \in A$, we define its \emph{left transformation} to be the following function of type $A \to A$:
      \[ x \mapsto \prod([a, \underline{x}])\]
      Observe that the set of all left transformations equipped with function composition forms a monoid. 
      This is because, thanks to the associativity axiom, the left behaviour of $\prod{[x, \underline{y}]}$ 
      is equal to the composition of left behaviours of $x$ and $y$. We define $M_L$ to be the monoid of left 
      transformations, and $h_L$ to be the function that maps elements of $A$ to their left behaviours. 
      Values $M_R$ and $h_L$ are defined analogously, but for \emph{right behaviours}.
      Observe that the value
      \[ \prod[a_1, \ldots, \underline{a_i}, \ldots, a_n] \]
      can be computed as $s(p(a_i))$, where $p$ is the $M_L$ product of the prefix, and $s$ is the $M_R$ 
      product of the suffix (as defined in the statement).  
\end{proof}

We are now ready to prove the equivalence of nondeterministic Mealy-machines and $\bar{L}$-definable transductions:
\begin{lemma}
    The class of transductions computed by unambiguous Mealy machines is equal to $\bar{L}$-definable transductions. 
\end{lemma}
\begin{proof}~

    $\subseteq$: A $\bar{L}$ definable transduction is given by a $\bar{L}$-algebra $(A, \alpha)$, 
    an input function $h : \Sigma \to S$ and an output function $\lambda : A \to \Gamma$.
    The $i$-th letter of the output is then computed as $\lambda(\alpha([h(w_1), \ldots, \underline{h(w_i)}, \ldots, w_n]])$.
    By \cref{lem:pointed-alg-prod-two-mon-prod}, we know that there are two monoids $M_R$, $M_L$, 
    and functions $h_R$, $h_L$, $g$, such that this can be computed as $g(p_L, h(w_i),s_R)$, 
    where $p_L = h'_L(w_1) \cdot \ldots s_R = h'_L(w_{i-1})$ and  $p_R = h'_R(w_{i+1}) \cdot \ldots \cdot h'_R(w_n)$, 
    for $h'_L = h_L \circ h$ and $h'_R = h_R \circ h$.  
    Based on this observation we can define an unambiguous Mealy machine. Intuitively
    the machine is going to remember in its state the $M_L$-product of the prefix and the $M_R$-product of the suffix, 
    and it is going to use $g$ to compute the output letters. Formally, the machine's set of states is equal to $M_L \times M_R$,
    its initial states are $\{1\} \times M_R$ and its final states are $M_L \times \{1\}$. Finally, its transition function 
    consists of the following tuples, for every $m_L \in M_L$, $m_R \in M_R$, and $x \in \Sigma$: 
    \[ (\underbrace{(m_L,\,  m_R \cdot h'_R(x))}_{\textrm{previous state}},
    \underbrace{x}_{\textrm{input letter}},
    \underbrace{(m_L \cdot h'_L(x),\, m_R)}_{\textrm{next state}},\
    \underbrace{g(m_L,\, h(x),\, m_R)}_{\textrm{output letter}}) \]
    Thanks to this definition, we know that the only correct run for an input $w \in \Sigma^*$ 
    is the run that correctly evaluates the  monoid products for all prefixes and suffixes.
    It follows that the machine is unambiguous, and correctly computes the output of the original $\bar{L}$-transduction.

    $\supseteq$: Now, we are given an unambiguous Mealy machine $\Sigma^+ \to \Gamma^+$ defined by some $(Q, I, F, \delta)$,
    and we show how to transform it into a $\bar{L}$-definable transduction. We start by defining the transition semigroup
    for the Mealy machine.
    It consists of behaviours,
    which are analogous to the deterministic behaviours from \cref{lem:mealy-overrightarrow-eq}, but are relations instead of functions, 
    and they ignore the output letter (in other words, it is the transition monoid for the underlying NFA, as defined e.g. in \cite[IV.3.2]{pin2010mathematical}):
    \[ M = \underbrace{Q}_{\substack{
        \textrm{The state in which}\\
        \textrm{the Mealy machine enters}\\
        \textrm{the infix from the left}
}} \times \underbrace{Q}_{\substack{
    \textrm{The state in which}\\
    \textrm{the Mealy machine exists}\\
    \textrm{the infix from the right} 
}}\]
The product operation in $M$ is simply the composition of relations: $f \cdot g = g \circ f$. Let us now use $M$ 
to define the $\bar{L}$-algebra $(A, \alpha)$. We start with the underlying set $A = M \times \Sigma \times M$. 
Before we define the product operation, let us show how to cast element of $A$ to $M$:
\[ t(m_1, a, m_2) = m_1 \cdot \delta(a) \cdot m_2 \comma \]
where $\delta(a) : Q \times Q$ is the partial application of the transition relation
(which computes the behaviour for the single letter $a$). We are now ready to define the product:
\[ \begin{tabular}{ccc} 
    $\alpha( [ (p_1, a_1, s_1), \ldots, \underline {(p_i, a_i, s_i)}, \ldots (p_n, a_n, s_n)])$\\
    $=$\\
    $\left(t(p_1, a_1, s_1) \cdot \ldots \cdot t(p_{i-1}, a_{i-1}, s_{i-1}) \cdot p_i, \ \ a_i,\ \ s_i \cdot t(p_{i+1}, a_{i+1}, s_{i+1}) \cdot t(p_n, a_n, s_n)\right)$
\end{tabular}
\]
It is not hard that this $\alpha$ satisfies the algebra axioms
(the singleton-mult axiom is straightforward, and the associativity axiom follows from the associativity of $M$).
Next, we define $h : \Sigma \to A$ as $h(a) = (1, a, 1)$.
Finally, we define $\lambda : A \to \Gamma$, in the following way: $\lambda(p, a, s) = b$, if there is a transition $q \transform{a/b} q'$,
an initial state $q_0$ and a final state $q_n$ such that $(q_0, q) \in p$ and $(q', q_n) \in s$. Thanks to the unambiguity
of the Mealy machine, we know that $\lambda$ is well-defined. We finish the proof by noting that the $\bar{L}$-transduction 
defined by $(A, \alpha)$, $h$ and $\lambda$ is by design equivalent to the original Mealy machine.     
\end{proof}

\section{Omitted details from \cref{sec:compositions-of-recognisable-transductions}}
\subsection{Naturality of $\putf$}\label{subsec:put-naturality}
The naturality of $\putf$ means that for every function $f : X \to Y$, the following diagram commutes (for the general
definition of naturality, see \cite[Definition~1.4.1]{abramsky2011introduction}
or \cite[Section~10]{milewski2018category}):
\[\begin{tikzcd}
	{M X \times X} && {M X} \\
	\\
	{M Y \times Y } && {M Y}
	\arrow["{\mathtt{put}}"{description}, from=1-1, to=1-3]
	\arrow["{(M f) \times f}"{description}, from=1-1, to=3-1]
	\arrow["{\mathtt{put}}"{description}, from=3-1, to=3-3]
	\arrow["{M f}"{description}, from=1-3, to=3-3]
\end{tikzcd}\]

\subsection{The flatten-expand axiom}\label{subsec:flatten-expand-ex}
In this section we give a step-by-step example of the flatten-expand axiom for $M = \overrightarrow{L}$, 
we start by restating the axiom:
\[\begin{tikzcd}
	& {M M M A} && {M M M A} \\
	\\
	{M M A} &&&& {M M A } \\
	&& {M A}
	\arrow["{\mu_A}"{description}, from=3-1, to=4-3]
	\arrow["{\delta_A}"{description}, from=4-3, to=3-5]
	\arrow["{\mu_{M A}}"{description}, from=1-4, to=3-5]
	\arrow["{\delta_{M A}}"{description}, from=3-1, to=1-2]
	\arrow["{M \mathtt{work}}", from=1-2, to=1-4]
\end{tikzcd}\]
where $\mathtt{work}$ is defined as the following composition:
\[ M M A \transform{\langle \id, \counit \rangle } M M A \times M A \transform{\id \times \delta} M M A \times M M A \transform{\mathtt{strength}} M (M M A \times M A)
\transform{M \putf} M M M A \transform { M \mu} M A \]

For the purpose of our example, let us consider $A = \{1, \ldots, 7\}$, and let us 
consider the following input $[[1, 2], [3, 4], [5, 6, 7]] \in \PL \PL A$.
The bottom path works as follows:
\[  \begin{tabular}{c}
    $[[1, 2], [3, 4], [5, 6, 7]]$ \\
    $\downmapsto\, \mu$\\
    $[1, 2, 3, 4, 5, 6]$\\
    $\downmapsto \, \delta$\\
    $[[1], [1, 2], [1, 2, 3], [1, 2, 3, 4], [1, 2, 3, 4, 5], [1, 2, 3, 4, 5, 6], [1, 2, 3, 4, 5, 6, 7]]$
\end{tabular}\]
Now let us focus on the top path. Here is the first step:
\[
    [[1, 2], [3, 4], [5, 6, 7]] \transformv{\delta } \left[[[1,2]],\ [[1, 2],\, [3, 4]],\ [[1, 2],\, [3, 4],\, [5, 6, 7]]\right]
\] 
The next step in the top path is $M\, \mathtt{work}$ which applies the work function in parallel 
to all elements of the top list. Let us show, how it works on the last element, i.e. on $[[1, 2], [3, 4], [5, 6, 7]]$: 
\[ \begin{tabular}{c}
    $\Big[[1, 2], [3, 4], [5, 6, 7]\Big]$\\
    $\downmapsto\, \langle id, \counit\rangle$\\
    $\left(\Big[[1, 2], [3, 4], [5, 6, 7]\Big], [5, 6, 7] \right)$\\
    $\downmapsto\, id \times \delta$ \\
    $\left(\Big[[1, 2], [3, 4], [5, 6, 7]\Big], \Big[[5], [5, 6], [5, 6, 7]\Big] \right)$\\
    $\downmapsto\, \strength$\\ 
    $ \left[ \biggl( \bigl[[1, 2], [3, 4], [5, 6, 7]\bigr], [5] \biggr),
             \biggl( \bigl[[1, 2], [3, 4], [5, 6, 7]\bigr], [5, 6] \biggr),
             \biggl( \bigl[[1, 2], [3, 4], [5, 6, 7]\bigr], [5, 6, 7] \biggr)
      \right] $\\
    $\downmapsto\, M \putf$ \\
    $\biggl[  \Bigl[[1, 2], [3, 4], [5]\Bigr], 
            \Bigl[[1, 2], [3, 4], [5, 6]\Bigr], 
            \Bigl[[1, 2], [3, 4], [5, 6, 7]\Bigr]
    \biggr]$\\
    $\downmapsto\, M \mu$ \\
    $\Bigl[ [1, 2, 3, 4, 5], [1, 2, 3, 4, 5, 6], [1, 2, 3, 4, 5, 6, 7]\Bigr]$
\end{tabular}
\]
In a similar fashion one can compute $\mathtt{work}$ for the other sublists:
\[ 
    \begin{tabular}{ccc}
    $[[1, 2]]$ & $\transformv{\mathtt{work}}$ & $[[1], [1, 2]]$\\
    $[[1, 2], [3, 4]]$ & $\transformv{\mathtt{work}}$ & $[[1, 2, 3], [1, 2, 3, 4]]$\\
    $[[1, 2], [3, 4], [5, 6, 7]]$ & $\transformv{\mathtt{work}}$ & $[ [1, 2, 3, 4, 5], [1, 2, 3, 4, 5, 6], [1, 2, 3, 4, 5, 6, 7]]$
    \end{tabular}
\]
Using those results of $\mathtt{work}$, we can trace the top path of the diagram:
\[ 
    \begin{tabular}{c}
    $\Bigl[[1, 2], [3, 4], [5, 6, 7]\Bigr]$\\ 
    $\downmapsto\, \delta$\\
    $\Bigl[ \bigl[[1,2]\bigr], \bigl[[1, 2], [3, 4]\bigr], \bigl[[1, 2], [3, 4], [5, 6, 7]\bigr] \Bigr]$\\
    $\downmapsto\, M \mathtt{work}$\\
    $\Bigl[ \bigl[[1],[1, 2] \bigr],\ \bigl[[1, 2, 3],[1, 2, 3, 4]\bigr],\ \bigl[[1, 2, 3, 4, 5], [1, 2, 3, 4, 5, 6], [1, 2, 3, 4, 5, 6, 7]\bigr]\Bigr]$\\
    $\downmapsto\, \mu$\\
    $\Bigl[[1], [1, 2], [1, 2, 3], [1, 2, 3, 4], [1, 2, 3, 4, 5], [1, 2, 3, 4, 5, 6], [1, 2, 3, 4, 5, 6, 7]\Bigr]$  
    \end{tabular}
\]
Which means that both in the top and in the bottom path, we obtain the same result. 

\subsection{Flatten-expand axiom as bialgebra}\label{subsec:flatten-expand-as-bialgebras}
In this section, we show how to describe the flatten-expand axiom in the language of \emph{bialgebras}. 
We start with the definition of \emph{coalgebras} for a comonad
(which is the dual of \cref{def:em-algabra}):
\begin{definition}
    A \emph{coalgebra} for a comonad $W$ is a set $S$ together 
    with a multiplication function $\beta : S \to W S$, that makes the
    following diagrams commute:
\[\begin{tikzcd}
	S && {WW S} & S && {W S } \\
	\\
	WS && {WW S} &&& {S }
	\arrow["\beta"{description}, from=1-1, to=1-3]
	\arrow["\beta"{description}, from=1-1, to=3-1]
	\arrow["{W \beta}"{description}, from=1-3, to=3-3]
	\arrow["\delta"{description}, from=3-1, to=3-3]
	\arrow["\beta"{description}, from=1-4, to=1-6]
	\arrow["\varepsilon"{description}, from=1-6, to=3-6]
	\arrow["id"{description}, from=1-4, to=3-6]
\end{tikzcd}\]
\end{definition}
We are now ready to present the definition of a bialgebra\footnote{
    Defined in \cite[Section~7.2]{turi1997towards}. 
    See also \cite{klin2011bialgebras}.
}:
\begin{definition}
    Let $M$ be a monad, and $W$ be a comonad. A $(M, W)$-bialgebra is a set 
    $S$ equipped with three functions:
    \[ 
        \begin{tabular}{ccc}
        $\alpha : M S \to S$ & $\beta : S \to M S$ & $\gamma : M W S \to W M S$
        \end{tabular}
    \]
    Such that $(S, \alpha)$ is an $M$-algebra, $(S, \beta)$ is a $W$-coalgebra, and the following diagram commutes:
    \[\begin{tikzcd}
	& {M W S} && {W M S} \\
	{M S} &&&& {W S} \\
	&& S
	\arrow["\alpha"{description}, from=2-1, to=3-3]
	\arrow["\beta"{description}, from=3-3, to=2-5]
	\arrow["{M \beta}"{description}, from=2-1, to=1-2]
	\arrow["{M \alpha}"{description}, from=1-4, to=2-5]
	\arrow["\gamma"{description}, from=1-2, to=1-4]
    \end{tikzcd}\]
\end{definition}
We are going to be intereseted in the case where the monad and the comonad is the same functor $M$.

Observe now, that for every set $X$, the set $M X$ equipped with the $\mu$ operation forms an $M$-algebra, 
called the \emph{free algebra} over $X$. (In this case the axioms of an algebra coincide with the axioms of a monad).
Similarly, the set $M X$ equipped with the $\delta$ operation, forms an $M$-coalgebra, called \emph{the free coalgebra}.

We are now ready to specify the flatten-expand axiom in terms of a bialgebra. It states that for
every $X$, the set $M X$ equipped with $\mu$ (i.e. the free algebra structure), $\delta$ (i.e. the free coalgebra structure),
and $\gamma$ defined as the following composition, is a bialgebra:
\[ \gamma : M M M X \transform{\delta} M M M M X \transform{M\langle M \counit, \counit \rangle} M (M M X \times M M X) \transform{M \strength} M M (M M X \times M X) \transform{M M \putf} \]
\[  \transform{M M \putf}  M M M M X \transform{\mu} M M M X\] 
After unfolding the bialgebra definition, this means that the following diagram commutes:
\[\begin{tikzcd}
	& {M M M X} && {M M M X} \\
	{M M X} &&&& {M M X} \\
	&& {M X}
	\arrow["\mu"{description}, from=2-1, to=3-3]
	\arrow["\delta"{description}, from=3-3, to=2-5]
	\arrow["{M \delta}"{description}, from=2-1, to=1-2]
	\arrow["{M \mu}"{description}, from=1-4, to=2-5]
	\arrow["\gamma"{description}, from=1-2, to=1-4]
\end{tikzcd}\]
Using basic equational reasoning, we can show that the top path in this diagram is equal to the 
top path in the flatten-expand axiom. (This is formalized as \texttt{flattenExpandAltThm}, see \cref{subsec:flatten-expand-bialg-coq}.)
It follows that the bialgebraic formulation is equivalent to the flatten-expand axiom.

\subsection{Omitted proofs from \cref{subsec:ctx}}
\begin{replemma}{\ref{lem:contexts-compose}}
    For every $f, g \in C_A$, it holds that $f \circ g \in C_A$. 
\end{replemma}
\begin{proof}
We start the proof by defining the following operation $\concat : M A \times M A \to M A$: 
\[ M A \times M A \transform{ M \eta_A \times \id } M M A \times M A \transform{\putf_{M A}} M M A \transform{\mu_A} M A\]
The intuition behind this operation is that it overrides the focused element, with the given element of $M A$.

For more intuition, consider the following example in $\overrightarrow{L}$
 and observe that the element $3$ disappears:
\[ \concat([1, 2, 3], [4, 5, 6]) = [1, 2, 4, 5, 6] \]
In the setting of lists this \emph{overriding} behaviour might seem counter-intuitive,
as there already is a more natural definition of concatenation. However, the overriding 
behaviour is clearly defined for all $M$s, and the usual definition of concatenation 
does not seem to be generalizable (for example for $M$'s such as $\bar{L}$ and $T_\mathcal{S}$).
For example, in $\bar{L}$, $\concat$ works as follows:
\[\concat([1, \underline 2, 3], [4, 5, \underline 6] ) = [1, 4, 5, \underline{6}, 2, 3] \]
Using the put-assoc axiom, one can show the context of concatenation is equal to the composition of contexts
(this is verified as \texttt{concatCtx} in Coq, see \cref{subsec:coq-contexts}):
\begin{lemma}\label{lem:contexts-concat}
    For every $k, l \in M A$, it holds that:
    \[ \ctx_k \circ \ctx_l = \ctx_{\concat(k, l)} \]
\end{lemma}
This finishes the proof of the \cref{lem:contexts-compose}.
\end{proof}

\subsection{Wreath Product}\label{subsec:wreath-product}

In this section, we show how to compose two $\overrightarrow{L}$-transductions using the usual definition
of a \emph{wreath product} \cite[Defnition~1.7]{krohn1965algebraic}. This serves two purposes: the first one
is to relate the \emph{generalized wreath product} from \cref{subsec:composition-theorem} with the classical wreath product, 
the second one is to give more intuitions about the proof of \cref{thm:compose}.

Remember that a $\overrightarrow{L}$-transduction of type $\Sigma^+ \to \Gamma^+$
is given by a semigroup $S$, and functions $h : \Sigma \to S$, $\lambda : S \to \Gamma$, 
and is computed according to the following formula:
\[ w_1 \ldots w_n \to \lambda(h(w_1)),\ \ldots,\ \lambda(h(w_1) \cdot \ldots \cdot h(w_n))\]

We are given two $\overrightarrow{L}$-transductions $F_1 : \Sigma^* \to \Gamma^*$ and $F_2 : \Gamma^* \to \Delta^*$, 
given by $(S_1, h_1, \lambda_1)$ and $(S_2, h_2, \lambda_2)$, and we would like to construct
$(S_3, h_3, \lambda_3)$, such that their $\overrightarrow{L}$-transduction computes the
composition $\Sigma^* \transform{F_1} \Gamma^* \transform{F_2} \Delta^*$. As mentioned before, 
for $S_3$, we are going to use the wreath product of $S_1$ and $S_2$. 

The intuition behind $S_3$ is that it represents the $S_1$- and $S_2$-products for every possible infix. 
Before we define $S_3$, let us show what it means to compute the $S_2$-product of an infix from $\Sigma^*$. 
For example, consider the following infix $w \in \Sigma^*$:
\[ \begin{tabular}{cccccc}
    \fbox{\textrm{unknown preffix}} & $a_1$ & $a_2$ & $a_3$ & $a_4$ &  \fbox{\textrm{unknown suffix}} 
\end{tabular}
\]
In order to compute the $S_2$-product of $w$, we need to first transform it using $F_1$.
This is slightly problematic, as the output of $F_1$ will usually depend on the unknown prefix that comes before $w$.
However, the only information we need about that prefix is its $S_1$-product. 
For example, if we know that $S_1$-product of the prefix is equal to $s$, we can compute the output 
of $F_1$ on $w$ as follows (where $s_i := h_1(a_i)$):
\[ \begin{tabular}{cccccc}
    \fbox{\textrm{unknown preffix}} & $\lambda_1(s \cdot s_1)$ & $\lambda_1(s \cdot s_1 \cdot s_2)$ & $\lambda_1(s \cdot s_1 \cdot s_2 \cdot s_3)$  & $\lambda_1(s \cdot s_1 \cdot s_2 \cdot s_3 \cdot s_4)$ &  \fbox{\textrm{unknown suffix}} 
\end{tabular}
\]
Now we can compute the $S_2$-product of $w$ by applying $h_2$ to each letter and multiplying the results (in $S_2$):
\[ h_2(\lambda_1(s \cdot s_1)) \cdot h_2(\lambda_1(s \cdot s_1 \cdot s_2)) \cdot   h_2(\lambda_1(s \cdot s_1 \cdot s_2 \cdot s_3)) \cdot  h_2(\lambda_1(s \cdot s_1 \cdot s_2 \cdot s_3 \cdot s_4))\]
According to this reasoning, the $S_2$ value of an infix can be represented as a function of the following type (where $S_1^1$ denotes $S_1$ adjoined with a formal identity element):
\[ \underbrace{S_1^1}_{ \substack{
    \textrm{Given the $S_1$-product of the preffix}\\
    \textrm{(where $1$ represents the empty preffix)}
   }} \to 
   \underbrace{S_2}_{\substack{
    \textrm{What is the $S_2$-product of the infix?}
   }}
\]
We are now ready to define $S_3$ as the following set:
\[ \underbrace{S_1}_{ \substack{
    \textrm{The $S_1$-value of the infix}
}} \times \underbrace{S_1^1 \to S_2}_{ \substack{
    \textrm{The $S_2$-value of the infix}\\
    \textrm{as explained above}
} } \]
The product operation on $S_3$ is defined with the following formula,
which follows from our definition of the $S_2$-values:
\[ (s, f) \cdot (t, h) = (s \cdot t,\ x \mapsto f(x) \cdot g(x \cdot s)) \]
It is now not hard to see that $S_3$ equipped with the following $h_3$ and $\lambda_3$ 
recognizes the composition $F_2 \circ F_1$: 
\[
    \begin{tabular}{l}
        $h(a) = \left(\; h_1(a),\ x \mapsto h_2\left(\lambda_1 \bigl(x \cdot h_1(a)\bigr)\right) \;\right)$\\
        $\lambda(x, f) = \lambda_2(f(1))$
    \end{tabular} 
\]

Let us finish this section by comparing this definition of a wreath product, with our definition 
of the \emph{generalized wreath product} from \cref{subsec:composition-theorem}, where it is defined as:
\[ S_1 \times (S_1^{S_1} \to S_2) \]
Remember that the set $S_1^{S_1}$ is meant to represent $S_1$-contexts, so we can think of this type as
$S_1 \times (C_{S_1} \to S_2)$. (Actually, this is how we could have defined the \emph{generalized wreath product} from the 
beginning. We used the more general definition for the sake of simplifying the definition of the product,
and the (formal) proof of associativity, but we believe that this finer definition would work as well).
Since, as explained in \cref{subsec:ctx}, for $M = \overrightarrow{L}$ the set of contexts $C_S$
is isomorphic to $S^1$, this can be further simplified to $S_1 \times (S_1^1 \to S_2)$,
which coincides with the definition of the \emph{wreath product} presented in this section.
    
\section{Omitted details from \cref{sec:futher-work}}
\subsection{General Cartesian closed categories}\label{subsec:ccc}
In this section we discuss possible strategies and obstacles for generalizing the results of this 
paper from \Set\ to arbitrary \emph{Cartesian closed categories}. 
We say that a category is called \emph{Cartesian closed} if it admits products
$X \times Y$ and function spaces $X \Rightarrow Y$ (see \cite[Definition~78]{abramsky2011introduction} for the full definition).

As it turns out, our \texttt{Coq}-formalization of the results mostly uses a subset of $\lambda$-calculus 
that can be automatically translated into morphism in every Cartesian closed category (see \cite[Section~1.6.5]{abramsky2011introduction}). 
The only exception is the function $\strength$, which is defined in the following way:
\[ \strength(x, l) = \left(M\,  (\lambda y . (x, y))\right)\  l \]
This causes problems, because in the Cartesian closed categories, the functor $M$ can only be applied to arrows of the category
(i.e. $X \to Y$) and not to the exponent objects (i.e. $X \Rightarrow Y$). (In particular, 
not every functor in every Cartesian closed category has a \emph{strength}.)

One way to deal with this problem is to require that the functor $M$ should come
equipped with a $\strength$ (smilarly to how it comes equipped with a $\putf$ function),
and axiomatize its expected behaviour.
We have tried this approach with the usual axioms of a \emph{strong functor}, 
\emph{strong monad}, and a \emph{strong comonad}\footnote{We took the axioms of a strong functor and a strong monad from \cite[Definition 3.2]{moggi1991notions}.
For the axioms of a strong comonad, we took the duals of the axioms for a strong monad.}, but we were not able to prove \cref{thm:compose}
within this axiomatization. Here is an example of a rather basic property, that does not seem to follow 
from this usual set of axioms:
\[\begin{tikzcd}
	{M X} && {(X\Rightarrow X) \times M X} && {M ( (X \Rightarrow X) \times X) } \\
	\\
	&&&& {M X}
	\arrow["id"{description}, from=1-1, to=3-5]
	\arrow["{\langle \mathtt{const_{id}}, id\rangle}", from=1-1, to=1-3]
	\arrow["{\mathtt{strength}}", from=1-3, to=1-5]
	\arrow["{M \mathtt{eval}}"{description}, from=1-5, to=3-5]
\end{tikzcd}\]
In the diagram $\mathtt{eval} : (Y \Rightarrow X) \times Y \to X$ denotes the function application 
(from the definition of a Cartesian closed category), and $\mathtt{const_{id}} : Z \to (X \Rightarrow X)$
denotes an arrow that maps every argument to the identity function
(formally, this is defined as $\mathtt{const_{id}} = \Lambda(\pi_2)$, where $\Lambda$ comes from the definition 
of a Cartesian closed category and $\pi_2 : Z \times X \to X$ is the second projection). 

Next, we tried adding this diagram as one of the axioms and proving \cref{thm:compose}. 
However we have encountered other problems. 
So, for the sake of simplicity, we have decided to restrict
the scope of this paper to \Set. However, we believe that it should be possible to find an axiomatization of
$\strength$ that would admit a proof of \cref{thm:compose}. Moreover, it is possible that 
such an axiomatization already exists in the literature and we were simply not able to find it.
We would welcome any suggestion of such an axiomatization. 

\section{Formalization in Coq}\label{subsec:coq}
In this section, we present the framework of our Coq formalization, focusing on key definitions and the statements of main lemmas.
To streamline our exposition, we exclude the formal proofs,
and certain auxiliary lemmas deemed peripheral to our core arguments.
The entire Coq file is available under the following link\footnote{
    \url{https://github.com/ravst/MonadsComonadsTransducersCoq}
}.

\subsection{Modelling}
In this section, we show how we have modelled our theory in Coq. 
We start by fixing the functor $M$:
\begin{lstlisting}[language=Coq]
Parameter M : Type -> Type.
Parameter mapM : forall {A} {B}, (A -> B) -> M A -> M B.
Notation " # f " := (mapM f).
\end{lstlisting}
Here \texttt{mapM} is the mapping on sets, and \texttt{mapM} is the mapping on functions.
We also introduce a notation for the mapping on sets, where $M f$ is written as $\mathtt{\#f}$.
Next, we assert the axioms of a functor:
\begin{lstlisting}[language=Coq]
Axiom mapCompose : forall {A} {B} {C} (f : B -> C) (g : A -> B) (mx : M A),
    (#f)((#g) mx) = (#(compose f g)) mx.

Axiom mapId : forall {A} (x : M A), (# id) x = x.
\end{lstlisting}
In a similar fashion, we assert the monad structure on $M$:
\begin{lstlisting}[language=Coq]
(*Flatten operation*)
Parameter mult : forall {A}, M (M A) -> M A. 
(*Singleton operation*)
Parameter unit : forall {A}, A -> M A.

(*Monad operations are natural*)
Axiom multNatural : forall A B, forall f : A -> B, forall x,
  mult((#(#f)) x) = (#f)(mult x).

Axiom unitNatural : forall A B, forall f : A -> B, forall x,
  unit(f x) = (#f)(unit x).

(*Monad operations satisfy monad axioms*)
Axiom multAx : forall A (x : M (M (M A))), 
  mult (mult x) = mult ((#mult) x).

Axiom multMapUnitAx : forall A (x : M A), 
  mult ((# unit) x) = x.

Axiom multUnitAx : forall A (x : M A),
  mult (unit x) = x.

\end{lstlisting}
Next, we assert the comonad structure on $M$:
\begin{lstlisting}[language=Coq]
(*Expand*)
Parameter coMult : forall {A}, M A -> (M (M A)).
(*Extract*)
Parameter coUnit : forall {A}, M A -> A.

(*Comonad operations are natural*)
Axiom coUnitNatural : forall A B, forall f : A -> B, forall x,
    f (coUnit x) = (coUnit ((#f) x)).

Axiom coMultNatural : forall A B, forall f : A -> B, forall m,
    coMult ((#f) m) = (#(#f)) (coMult m).

(*Comonad operations satisfy comonad axioms*)
Axiom coMultAx : forall A (x : M A),
    coMult (coMult x) = (#coMult) (coMult x).

Axiom coUnitCoMultAx : forall A (x : M A),
    coUnit (coMult x) = x.

Axiom mapCoUnitComultAx : forall A (x : M A),
    (#coUnit) (coMult x) = x.
\end{lstlisting}
Next, we introduce the $\putf$ operation, and assert that it is natural:
\begin{lstlisting}[language=Coq]
Parameter put : forall {A},  ((M A) * A) -> M A.

(*Put is natural*)
Axiom putNatural : forall {A} {B} (f : A -> B) (xs : M A) (x : A),
  (#f) (put (xs, x)) = put ((#f) xs, f x).
\end{lstlisting}
Then, we introduce the coherence axioms. In the following code, we use the notation 
$\mathtt{<\!| f,\ g|\!>}$ for the pairing of two functions $\langle f, g \rangle$, defined 
as $\langle f, g \rangle x = (f x, g x)$, and the notation $\mathtt{*}$ for function composition.
\begin{lstlisting}[language=Coq]
Axiom flattenExtract : forall A, forall (x : M (M A)),
    coUnit (mult x) = coUnit (coUnit x).

Axiom singletonExpand : forall A, forall (x : A), 
    coMult (unit x) = (#unit) (unit x).
  
Axiom singletonExtract : forall A, forall (x : A), 
    coUnit (unit x) = x.

Axiom getPut : forall A, forall (x : M A) (y : A),
    coUnit (put (x, y)) = y.

Axiom putGet : forall A, forall (x : M A), 
    put (x, coUnit x) = x.

Axiom putPut : forall {A} l (x : A) y, 
    put (put(l, x), y) = put (l, y).

Axiom putAssoc : forall A, forall (x : M (M A)) (y: M A) (z : A),
    put ((mult (put (x, y))), z) = mult (put (x, put (y, z))).

Axiom singletonPut : forall A (x : A) y, 
    put(unit x, y) = unit y. 

(*The Set-specific definition of strength*)
Definition str {X} {Y} (x : X * (M Y)) : (M (X*Y)) := 
match x with 
    (x1, x2) => (#(fun y => (x1, y))) x2
end.

Axiom flattenExpand : forall {A} (x : M (M A)),
    coMult (mult x) = 
    mult ( (#(# mult)) ( ((#(#put)) ((# str) ((#<| id, coMult * coUnit|>) (coMult x)))))).
\end{lstlisting}
Next, we define the properties of an algebra:
\begin{lstlisting}[language=Coq]
Definition associative {S} (alpha : M S -> S) : Prop := 
    forall l,  alpha ((#alpha) l) = alpha (mult l).
Definition unitInvariant {S} (alpha : M S -> S) : Prop := 
    forall s, alpha (unit s) = s.
\end{lstlisting}
Finally, we define an $M$-definable transduction:
\begin{lstlisting}[language=Coq]
Definition mTransduction {X Y S} (alpha : M S -> S) (h : X -> S) (lambda : S -> Y) : M X -> M Y := 
    (#lambda) * (#alpha) * coMult * (#h).
\end{lstlisting}

\subsection{The composition theorem}\label{subsec:composition-theorem-coq}
In this section, we present the formal statement of the composition theorem. 
We start with the context: two $M$-definable transduction $F : M X \to M Y$ 
and $G : M Y \to  M Z$:
\begin{lstlisting}[language=Coq]
(*We are given three alphabets *)
Variable X : Set.
Variable Y : Set.
Variable Z : Set.

(*We are given and M-transduction F :  M X -> M Y*)
Variable S1 : Set.
Variable prod1 : M S1 -> S1.
Variable h1 : X -> S1.
Variable lambda1 : S1 -> Y.
Axiom assoc1 : associative prod1.
Axiom unitInvariant : unitInvariant prod1.
Definition F := mTransduction prod1 h1 lambda1.

(*And we are given an M-transduction G : MY -> M Z*)
Variable S2 : Set.
Variable prod2 : M S2 -> S2.
Variable h2 : Y -> S2.
Variable lambda2 : S2 -> Z.
Axiom assoc2 : associative prod2.
Axiom unitInvariant2 : unitInvariant prod2.
Definition G := mTransduction prod2 h2 lambda2.
\end{lstlisting}
Next, we define the generalized wreath product of $S1$ and $S2$, 
and use it to define a new $M$-definable transduction $GF : M X \to M Z$:
\begin{lstlisting}[language=Coq]
Definition S3 : Type := S1 * ((S1 -> S1) -> S2).

Definition prod3 (l : M S3) : S3 :=
  let ctx1 (l : M S1) (x : S1) : S1 := prod1 (put (l, x)) in 
  let tmp1 (l : M S3) : ((S1 -> S1) -> S2) := proj2 (coUnit l) in 
  let tmp2 (c : S1 -> S1) (l : M S3) : (S1 -> S1) := c * (ctx1 ((#proj1)(l))) in 
  (prod1 ((#proj1) (l)), fun c => prod2 ( (#app2) (((#<|tmp1, tmp2 c |>) (coMult l))))).

Definition h3 (x : X) : S3 := 
  (h1 x, fun c => h2 (lambda1 (c (h1 x)))).

Definition lambda3 (s : S3) : Z :=
  match s with 
    | (_, f) => lambda2 (f (fun a => a))
  end.
Definition GF := mTransduction prod3 h3 lambda3.
\end{lstlisting}

Next, we prove that $GF$ is equal to $G \circ F$. 
(As mentioned before, we omit the proof in this paper, but it is available in the Coq file.)
\begin{lstlisting}[language=Coq]
Theorem compositionCorrect : GF = G * F.
\end{lstlisting}
Finally, we prove that $S3$ is a valid algebra (again, we omit the proofs):
\begin{lstlisting}[language=Coq]
Theorem S3Associative : associative prod3.
Theorem S3UnitInvariant : unitInvariant prod3.
\end{lstlisting}

\subsection{Contexts}\label{subsec:coq-contexts}
In this section, we present the formalization of the results from \cref{subsec:ctx}. 
We start with an algebra $S$: 
\begin{lstlisting}[language=Coq]
Variable S : Set.
Variable prod : M S -> S. 
Axiom prodAssoc : associative prod.
Axiom prodUnit : unitInvariant prod.
\end{lstlisting}
We define contexts: 
\begin{lstlisting}[language=Coq]
Definition ctx (l : M S) (x : S) : S := prod (put (l, x)).
\end{lstlisting}
And we prove the required lemmas, respectively Lemmas~\ref{lem:ctx-put}, \ref{lem:ctx-id}, and \ref{lem:contexts-concat}.
Here $\mathtt{<\!* f,\ g*\!>}$ denotes the function $f \times g$, defined as $(f \times g)(x, y) = (f x,\ g y)$.
\begin{lstlisting}[language=Coq]
Lemma ctxPutInvariant : forall l a, ctx l = ctx (put (l, a)).

Lemma ctxUnitId : forall x, ctx (unit x) = id.

Definition concat : M S * M S -> M S := 
  mult * put * <* (#unit), id *>.

Lemma concatCtx : forall (v w : M S), 
  ctx v * ctx w = ctx (concat (v, w)).
\end{lstlisting}
\subsection{Flatten-expand axiom}\label{subsec:flatten-expand-bialg-coq}
Finally, we formalize the equivalence of the flatten-expand axiom and the bialgebraic formulation
(see \cref{subsec:flatten-expand-as-bialgebras}).
The left-hand side of the equality is the top path in the flatten-expand axiom,
and the right-hand side is the top path in the bialgebraic formulation:
\begin{lstlisting}[language=Coq]
Theorem flattenExpandAltThm : forall A, forall (x : M (M A)), 
mult ( (#(# mult)) ( ((#(#put)) ((# str) ((#<| id, coMult * coUnit|>) (coMult x)))))) = 
((#mult) * mult * (#(# put)) * (# str) * (#<| #coUnit, coUnit|>) * coMult * (#coMult)) x.
\end{lstlisting}

\end{document}